%% file: Journal.tex
\documentclass[journal]{IEEEtran}
\input{preamble}

\usepackage{orcidlink}
\usepackage{hyperref}
\hypersetup{
    colorlinks=true,
    linkcolor=blue,
    citecolor=blue,
    urlcolor=black
}

\title{\centering {Temperature-aware Optimization of Liquid Crystal Reconfigurable Intelligent Surfaces:\\ Physics-based Modeling and Robust Design}}
\author{
\IEEEauthorblockN{Mohamadreza Delbari$^{\orcidlink{0000-0002-4768-5874}}$, Bowu Wang$^{\orcidlink{0009-0009-6037-0015}}$, 
Arash Asadi$^{\orcidlink{0000-0001-9946-4793}}$, and Vahid Jamali$^{\orcidlink{0000-0003-3920-7415}}$}
\IEEEauthorblockA{
\thanks{The work of M. Delbari, B. Wang, and V. Jamali was supported in part by the Deutsche Forschungsgemeinschaft (DFG, German Research Foundation) within the Collaborative Research Center Multi-Mechanisms Adaptation for the Future Internet (MAKI) (SFB 1053) under Project-ID 210487104; in part by the LOEWE initiative (Hesse, Germany) within the emergenCITY Centre under Grant LOEWE/1/12/519/03/05.001(0016)/72, and in part by the German Federal Ministry for Research, Technology and Space (BMFTR) under the program of ``Souverän. Digital. Vernetzt.'' joint project Open6GHub plus (Project-ID 16KIS2407). 
A. Asadi's work was in part supported by DFG HyRIS (455077022) and DFG mmCell (416765679). 
This paper was presented in part at the 2025 IEEE International Conference on Communications (ICC) [DOI: 10.1109/ICC52391.2025.11161472] in \cite{delbari2024temperature}. (\textit{Corresponding author: Mohamadreza Delbari.})\\
Mohamadreza Delbari, Bowu Wang, and Vahid Jamali are with the Resilient Communication Systems Laboratory, Technische Universität Darmstadt, 64283 Darmstadt, Germany (e-mail: mohamadreza.delbari@tu-darmstadt.de;
bowuwang98@gmail.com;
vahid.jamali@tu-darmstadt.de).
Arash Asadi is with the Embedded Systems Group, Delft University of Technology, 2628 CD Delft, Netherlands (e-mail: a.asadi@tudelft.nl).}}
}
\begin{document}
\maketitle

\begin{abstract}
While \gls{LC} technology facilitates the realization of energy-efficient and scalable \glspl{RIS}, their phase shift response is inherently temperature-dependent. Neglecting this thermal dependency can lead to performance degradation, which is particularly detrimental in secure wireless systems where phase-shift inaccuracies may result in unintended information leakage. To address this challenge, we investigate secure communication in \gls{LC}-\gls{RIS}-aided systems and develop a temperature-adaptive phase-shift design. Beyond thermal sensitivity, the massive number of elements at \gls{mmWave} frequencies is required to compensate for high path loss. This large-scale deployment of \gls{LC}-\glspl{RIS} can lead to significant overhead challenges due to the acquisition of \gls{CSI}. To ensure practical feasibility, this work proposes a phase-shift design that does not rely on the full \gls{CSI}; instead, it employs only the possible locations of legitimate users and potential eavesdroppers. By illuminating a spatial zone rather than a single target location, the proposed temperature-adaptive algorithm enhances robustness against both thermally induced phase errors and positioning inaccuracies. To solve the resulting optimization problem, we present a \gls{SDP}-based approach to serve as a high-performance benchmark, as well as a low-complexity heuristic method. The latter demonstrates superior scalability as the number of \gls{RIS} elements increases, which makes it highly effective for deploying extremely large surfaces in dynamic, real-time environments. Based on this scalable framework, we further design a temperature-robust algorithm that maintains high security without requiring real-time temperature data. Extensive simulation results confirm that our temperature-adaptive and temperature-robust approaches yield a superior secrecy rate compared to conventional designs that neglect temperature impacts.

\end{abstract}
\glsresetall
\glslocalunset{RIS}

\section{Introduction}
\label{sec Introduction}
\IEEEPARstart{R}{econfigurable} intelligent surfaces (RISs) are a potential technology for the next generation of wireless communications, with the vision of realizing programmable radio environments \cite{Wu2019,di2019smart,najafi2020physics,delbari2026fast}. \Gls{LC} technology has recently been studied as a cost-effective and energy-efficient solution for \gls{RIS} implementation, particularly for \gls{mmWave} communication systems \cite{zografopoulos2019liquid,aboagye2022design}. \Gls{LC}s and \gls{LC}-\gls{RIS}s have been investigated in the literature from both experimental and theoretical perspectives  \cite{aboagye2022design,zografopoulos2019liquid,neuder2023compact,jimenez2023reconfigurable,wang2004correlations,Wang2005,delbari2024fast,delbari2026wideband,gholian2025temperature}. For example, Neuder \textit{et al.} \cite{neuder2023compact} demonstrated an experimental design of an \gls{LC}-\gls{RIS}, while Aboagye \textit{et al.} \cite{aboagye2022design} focused on its applications in visible light communication. Additionally, Jim{\'e}nez-S{\'a}ez \textit{et al.} \cite{jimenez2023reconfigurable} provided a comprehensive review of key characteristics of \gls{LC}-\gls{RIS}, including power consumption and cost, and compared these with other related technologies. Based on the works by Wang \textit{et al.} \cite{wang2004correlations,Wang2005}, where the equations for the response time of liquid crystals were derived, the authors in \cite{delbari2024fast} formulated an optimization problem to reduce the switching time response of \gls{LC}-\gls{RIS} systems. The wideband application of the \gls{LC}-\gls{RIS} is also analyzed in \cite{delbari2026wideband}.  These works highlight the growing focus on addressing the efficiency and practicality of \gls{LC} technologies in \gls{RIS}-assisted systems.

The phase-shift responses of \gls{LC}-based \gls{RIS}s are inherently temperature-dependent, as they rely on the mechanical reorientation of \gls{LC} molecules to produce different phase shifts. While this temperature dependency typically has only a negligible impact on the main lobe, it can substantially distort the side lobes, which causes significant performance degradation, particularly in the context of physical layer security, due to unintended information leakage. \Gls{RIS}-assisted physical layer security has attracted substantial research interest \cite{Shen2019,Asaad2022,Chu2021,Xiu2021,huang2025secrecy}. Initial efforts focused on techniques known as frequently incorporating artificial noise injection or reflective phase adjustments to impair eavesdropper links \cite{Chu2021}, where techniques like alternating optimization and fractional programming were deployed to maximize weighted secrecy sum-rates \cite{Asaad2022}. Subsequent literature expanded into hardware-constrained environments; for example, \gls{mmWave} frameworks operating with low-resolution digital-to-analog converters have utilized block coordinate descent and successive convex approximation to counter hardware impairments during phase shift optimization \cite{Xiu2021}. More recently, researchers have shifted attention toward practical channel uncertainties. To combat imperfect \gls{CSI} across multi-user and multi-eavesdropper topologies, robust frameworks combining \gls{SDP} with iterative hybrid algorithms have been exploited to guarantee secure transmission rates under estimation errors \cite{huang2025secrecy}. Despite these prior works, to the best of the authors' knowledge, physical layer security in \gls{LC}-\gls{RIS}-assisted systems and the impact of temperature variations on its performance have not been studied in the existing literature.

This paper investigates the impact of temperature fluctuations on \gls{LC}-\gls{RIS} unit cells in secure communications that focuses on mitigating temperature-induced information leakage. Specifically, we optimize the \gls{LC}-\gls{RIS} phase shifts under joint location uncertainty for both the legitimate user and the eavesdropper. Because \gls{mmWave} systems suffer from severe path loss, a massive number of \gls{RIS} elements is required. This makes full \gls{CSI} acquisition impractical. To counter potential positioning errors and user mobility overhead, we design the \gls{LC}-\gls{RIS} phase shifts to cover a specific spatial area rather than a single point. This approach effectively minimizes pilot overhead from user movements while protecting against eavesdropper location uncertainty. Our main contributions are summarized below:

\begin{figure}
    \centering
    \includegraphics[width=0.5\textwidth]{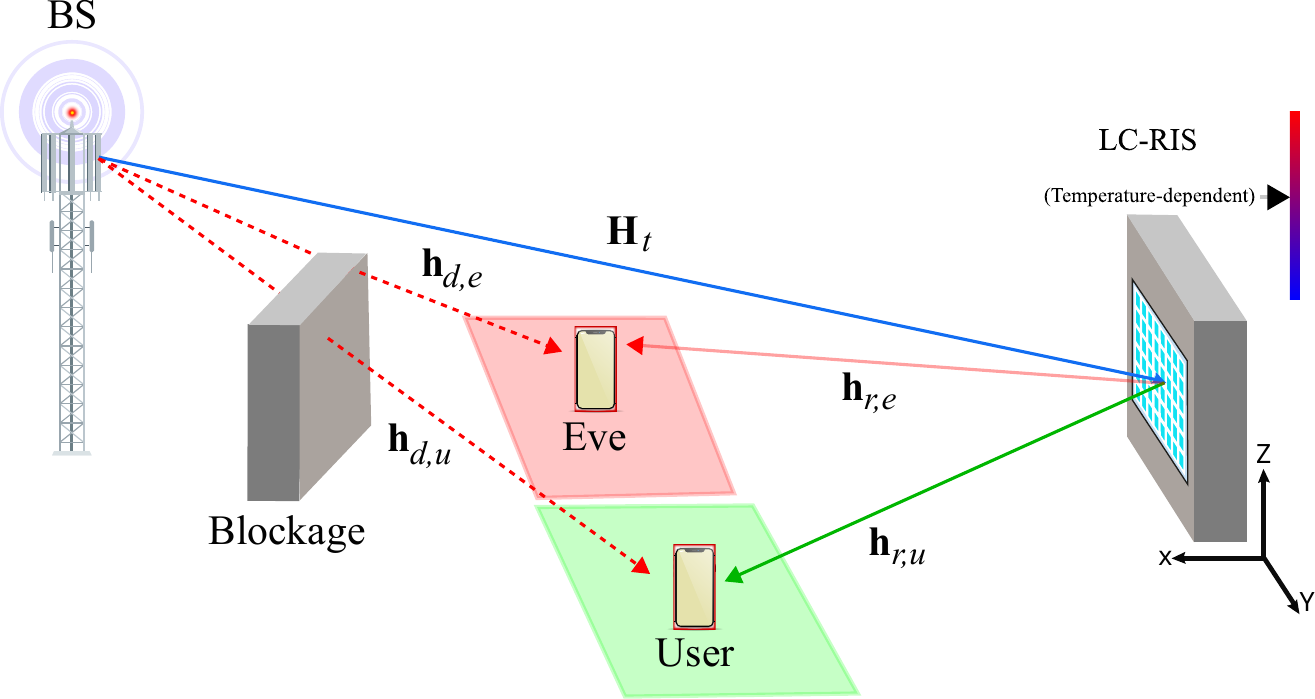}
    \caption{A wireless channel model where a \gls{BS} serves a legitimate user via an \gls{RIS} trying to decrease the received signal by an eavesdropper.}
    \label{fig:system model}
    \vspace{-5mm}
\end{figure}

\begin{itemize}
    \item \textbf{Physics-based Characterization of \gls{LC}-\gls{RIS} Phase-Shift Thermal Dynamics:} First, we introduce a physics-based model to quantify how temperature fluctuations impact the phase-shift profile of \gls{LC}-\gls{RIS} elements. Specifically, the model determines the phase shift at any arbitrary temperature relative to a known baseline phase shift at a standard temperature. We demonstrate that thermal variations alter the achievable differential phase-shift range $[0, \Delta\omega_{\max}]$, where $\Delta\omega_{\max}$ falls below $2\pi$ at temperatures higher than the standard baseline.
    \item \textbf{Temperature-aware Problem Formulation for \gls{LC}-\gls{RIS} Phase-shift Design:} Next, we introduce a physical layer security system model consisting of a legitimate \gls{MU} and a \gls{ME}. We formulate two non-convex optimization problems in Sections~\ref{sec: Temperature-Aware LC-RIS Phase-shift Design} and \ref{sec: Robust Design} to design the \gls{LC}-\gls{RIS} phase shifts, one that adapts to the temperature (temperature-adaptive design) and one that is robust to temperature variations (temperature-robust design). For both problems, we assume that the exact locations and instantaneous \gls{CSI} for \gls{MU} and \gls{ME} are unknown, and we rely solely on approximate user vicinity data. This formulation captures a practical challenge of securing \gls{RIS}-aided communications in real-world environments where users are mobile and location estimation errors are unavoidable. 
\item \textbf{Proposed Algorithms:} To handle the high non-convexity of the optimization problems, we employ algebraic reformulations and introduce an efficient, near-optimal design based on \gls{SDP}. While this \gls{SDP} approach delivers high performance, its computational overhead scales cubically with the \gls{RIS} dimension. Therefore, we additionally develop a low-complexity alternative which is used for both temperature-adaptive and temperature-robust designs. This alternative scheme reduces the computational complexity, which makes it highly scalable for massive \gls{RIS} deployments.
    \item \textbf{Performance Evaluation:} Finally, we validate the proposed algorithms through extensive simulations. First, we evaluate the convergence behavior and empirically analyze the computational complexity of both methods by comparing their execution times. Next, we illustrate the performance gap between the two algorithms. For the scalable approach, we demonstrate the achievable secrecy rate across different temperatures, and prove that neglecting thermal effects during the design phase degrades performance. Finally, we show that even without prior knowledge of the operating temperature, our proposed robust design delivers a consistently high secrecy rate across all temperature profiles.
\end{itemize}
This paper significantly extends its conference version \cite{delbari2024temperature} in several directions. First, unlike \cite{delbari2024temperature}, we rigorously derive a physics-based model for the temperature dependency of the \gls{LC}-\gls{RIS} phase shift response. Second, while \cite{delbari2024temperature} only derived an \gls{SDP} solution, here we derive a scalable solution for temperature-adaptive design, which is necessary for extremely large \glspl{RIS}. Third, we propose a temperature-robust design not covered in \cite{delbari2024temperature}. Finally, extensive simulation results are provided to evaluate the performance. The remainder of this paper is organized as follows. In Section~\ref{sec: System, LC, and Secrecy rate Model}, we present the system, channel, and secrecy model. In Section~\ref{sec: LC phase shifter}, we detail the \gls{LC} model with focusing on the temperature impacts. Sections~\ref{sec: Temperature-Aware LC-RIS Phase-shift Design} and \ref{sec: Robust Design} detail the proposed optimization problem for temperature-adaptive and -robust designs, respectively, followed by simulation results in Section~\ref{sec: Performance Comparison}. Finally, Section~\ref{sec: Conclusion} concludes the paper.

\textit{Notation:} Bold capital and small letters are used to denote matrices and vectors, respectively.  $(\cdot)^\Trans$, $(\cdot)^\Herm$, $\rank(\cdot)$, and $\tr(\cdot)$ denote the transpose, Hermitian, rank, and trace of a matrix, respectively. Moreover, $\diag(\bA)$ is a vector that contains the main diagonal entries of matrix $\bA$, $\bone_n$ and $\bzero_n$ denote column vectors of size $n$ whose elements are all ones and zeros, respectively. $\|\bA\|_*=\sum_i \sigma_i$, $\|\bA\|_2=\max_i \sigma_i$, $\|\bA\|_F$, and $\blambda_{\max}(\bA)$ denote the respectively nuclear, spectral, and Frobenius norms of a matrix $\bA$, and eigenvector associated with the maximum eigenvalue of matrix $\bA$, where $\sigma_i,\,\,\forall i$, are the singular values of $\bA$. Furthermore, $[\bA]_{m,n}$ and $[\ba]_{n}$ denote the element in the $m$th row and $n$th column of matrix $\bA$ and the $n$th entry of vector $\ba$, respectively. $x^+$ denotes as $\max\{x,0\}$ and $\arg(\cdot)$ returns phase of a complex number between 0 and $2\pi$. Moreover, $\Rset$ and $\Cset$ represent the sets of real and complex numbers, respectively, $\jj$ is the imaginary unit, and $\Ex\{\cdot\}$  represents expectation. $\mathrm{rand}(N)$ denotes an $N\times1$ vector where each element is generated independently and uniformly from 0 to 1. $\mathcal{CN}(\bmu,\bSigma)$ denotes a complex Gaussian random vector with mean vector $\bmu$ and covariance matrix $\bSigma$. Finally, $\bigO(\cdot)$ is the big-O notation and $|\Pset|$ is the cardinality of set $\Pset$.

\section{System, Channel, and Secrecy Rate Model}
\label{sec: System, LC, and Secrecy rate Model}
In this section, we begin by presenting the system model for both the legitimate user and mobile eavesdropper. Subsequently, we describe the channel model used in this paper. Finally, we introduce the secrecy rate considered in this paper.

\subsection{System Model}\label{system model}This work considers a downlink communication system operating over a frequency-non-selective (narrow-band) channel. The system comprises a \gls{BS} equipped with $N_t$ \gls{Tx} antennas, which serves a single-antenna legitimate \gls{MU} in the presence of a single-antenna \gls{ME}\footnote{To focus on the temperature impacts, we consider a single \gls{MU} and \gls{ME}. However, the algorithms can be used for multi-user scenarios, which will lead to a more involved optimization problem.}. The transmission is assisted by an \gls{LC}-based \gls{RIS} consisting of $N$ unit cells. The signals received at the legitimate \gls{MU} and the \gls{ME}, denoted by $y_u \in \Cset$ and $y_e \in \Cset$ respectively, in temperature $T$ are expressed as:
\begin{align}
\label{Eq:system model user}
y_g = \big(\bh_{d,g}^\Herm + \bh_{r,g}^\Herm \bGamma(T) \bH_t \big) \bx +n_g,\quad g=\{u,e\},
\end{align}
where $\bx \in \Cset^{N_t}$ represents the transmit signal vector, and $n_g \sim \sCN(0,\sigma_n^2)$ denotes the \gls{AWGN} with power $\sigma_n^2$. We employ linear beamforming such that $\bx = \bq s$, where $\bq \in \Cset^{N_t}$ is the \gls{BS} beamforming vector and $s$ is the data symbol with unit average power, i.e., $\Ex\{|s|^2\}=1$. The beamformer is subject to a maximum power constraint $\|\bq\|^2 \leq P_t$. The \gls{RIS} response is characterized by the diagonal reflection matrix $\bGamma(T) \in \Cset^{N \times N}$, where the $n$-th diagonal element is defined as $[\bGamma(T)]_n = [\bOmega]_n e^{\jj[\bomega(T)]_n}$. Here, $[\bomega(T)]_n$ and $[\bOmega]_n$ represent the phase shift and reflection amplitude of the $n$-th unit cell, respectively. In accordance with established characteristics of \gls{LC}-\glspl{RIS} in narrow-band operation \cite{yang2020design}, amplitude variations are negligible, allowing us to assume $[\bOmega]_n \approx 1$ for all $n$\footnote{We assume the \gls{LC} elements' loss is absorbed in the required \gls{BS} transmit power \cite{Wang2025}.}. The channel matrices for the \gls{BS}-\{MU, ME\}, \gls{BS}-\gls{RIS}, and \gls{RIS}-\{MU, ME\} links are denoted by $\bh_{d,g}$, $\bH_t$, and $\bh_{r,g}, g=\{u,e\},$ respectively. A rigorous characterization of these links follows in the next subsection.

\subsection{Channel Model}
\label{sec: Channel Model}
Given that \gls{mmWave} \gls{RIS} deployments are typically positioned at elevated altitudes to mitigate ground blockages, the propagation environment is predominantly characterized by \gls{LOS} components rather than \gls{nLOS} scattering for \gls{BS}-\gls{RIS} and \gls{RIS}-\gls{MU} channels. Furthermore, the substantial physical dimensions of \gls{LC}-\gls{RIS} arrays often place the \gls{BS} and \glspl{MU} within its radiating \gls{NF} region. To accurately capture these effects, we adopt the generalized \gls{NF} \gls{MIMO} Rician model proposed in \cite{delbari2025near}, which accounts for the spherical wave propagation inherent in \gls{NF} scenarios \cite{Liu2023nearfield,delbari2024nearfield}. A general \gls{MIMO} channel $\bH \in \Cset^{N_\rx \times N_\tx}$ between $N_\tx$ transmit and $N_\rx$ receive antennas is modeled as:
\begin{equation}\label{eq: channel model}\bH=c_0(\bH^{\mathrm{LOS}}+\sum_{r=1}^R\bar{k}_r\bar{\bH}_r+\tilde{k}_r\tilde{\bH}_r),
\end{equation}
where $\bH^\LOS$ represents the direct \gls{NF} \gls{LOS} component. The terms $\bar{\bH}_r$ and $\tilde{\bH}_r$ denote the deterministic and stochastic \gls{nLOS} components associated with the $r$-th reflector (e.g., environmental surfaces like walls or the ground). The channel is parameterized by the \gls{LOS} amplitude $c_0$, and the deterministic and stochastic Rician $K$-factors, $\bar{k}_r$ and $\tilde{k}_r$, respectively. In the \gls{NF} regime, the deterministic components are modeled based on the antenna coordinates to account for a spherical wavefront:
\begin{subequations}
\label{eq: channel models}
\begin{align}
	[\bH^\LOS]_{m,n} &= \, \e^{\jj\kk\|\bu_{\rx,m}-\bu_{\tx,n}\|}\label{Eq:LoSnear},\\
    [\bar{\bH}_r]_{m,n} &= \, \e^{\jj\kk\|\bu_{\rx,m}^r-\bu_{\tx,n}\|},\\
    [\tilde{\bH}_r]_{m,n} &\sim\sCN(0,1),
\end{align}
\end{subequations}
where $\bu_{\rx,m}$ and $\bu_{\tx,n}$ denote the spatial coordinates of the $m$-th receive and $n$-th transmit elements, respectively, and $\bu_{\rx,m}^r$ represents the image point relative to the $r$-th reflecting surface. The wavenumber is defined as $\kk=2\pi f/c$ where $f$ and $c$ are the frequency and the speed of light in vacuum, respectively. This unified model is applied to the channels $\bH_t$ and $\bh_{r,g},\,\forall g\in\{u,e\}$. To reflect practical constraints, we incorporate a barrier penetration model for the \gls{BS}-\glspl{MU} \gls{LOS} link to account for potential blockages \cite{Phillips2013}. For the purposes of algorithm design in Sections~\ref{sec: Temperature-Aware LC-RIS Phase-shift Design} and \ref{sec: Robust Design}, we assume the direct link is fully obstructed, i.e., $\bh_{d,g} \approx \bzero_{N_t},\,\forall g\in\{u,e\}$. However, the contribution of the direct link is considered in the numerical evaluations presented in Section~\ref{sec: Performance Comparison} to assess the system's performance under more general conditions.

\subsection{Secrecy Rate}
\label{sec: Secrecy rate}
To evaluate the effectiveness of physical-layer security, the secrecy rate is used as a primary performance metric, which quantifies the difference between the transmission rate accessible to legitimate \gls{MU} and the rate intercepted by an unauthorized \gls{ME}. The secrecy rate is defined as \cite{Cheng2023secure}:
\begin{subequations}
    \label{eq: SNR}
\begin{align}
\label{eq: SNR a}
    \RS(T)&=[\widetilde{\RS}(T)]^+,\\
\label{eq: SNR b}
    \widetilde{\RS}(T) &=\log(1\!+\!\SNR_u(T))\!-\!\log(1\!+\!\SNR_e(T)),
    \end{align}
\end{subequations}
where
\begin{subequations}
\begin{align}
    \SNR_u(T)&=|(\bh_u^\eff(T))^\Herm\bq|^2/\sigma^2_n,\\
    \SNR_e(T)&=|(\bh_e^\eff(T))^\Herm\bq|^2/\sigma^2_n.
\end{align}
\end{subequations}
Here, $(\bh_g^\eff(T))^\Herm=\bh_{d,g}^\Herm + \bh_{r,g}^\Herm \bGamma(T) \bH_t,\, \forall g \in \{u,e\}$ denotes the total effective channel and encompasses the combined communication link from the \gls{BS} to the legitimate \gls{MU} and the \gls{ME}, which accounts for both direct and reflected paths.

This paper focuses on maximizing the secrecy rate, $\RS(T)$, by accounting for temperature-dependent variations on \gls{LC} and jointly optimizing the \gls{BS} beamformer and \gls{RIS} phase shifts. This enhances the signal for the legitimate user while intentionally degrading the signal quality for potential eavesdroppers. Unlike traditional models that require precise \gls{CSI}, our approach utilizes the approximate location regions of both parties, which offers several practical benefits:

\begin{itemize}
    \item \textbf{Robust design \gls{w.r.t.} eavesdropper's channel knowledge:} Rather than assuming full (instantaneous or statistical) \gls{CSI} for the eavesdropper, which is a common but often difficult to realize \cite{Cumanan2014}, we only require that their location be confined to a general spatial region, $\bp_e \in \Pset_e$. The physical dimensions of this zone, $|\Pset_e|$, naturally account for estimation uncertainties.
    \item \textbf{Reduced overhead and increased user coverage:} Similarly, legitimate \glspl{MU} are assumed to be located within a target zone, $\bp_u \in \Pset_u$. Increasing the area $|\Pset_u|$ serves a dual purpose: it compensates for localization errors and minimizes control overhead by reducing the \gls{RIS} reconfigurations frequency, but at the cost of reducing receive power. Therefore, the \gls{RIS} phase shifts are optimized to provide consistent service across all potential user coordinates based on the given temperature knowledge.
\end{itemize}

In this paper, the design of the \gls{LC}-\gls{RIS} phase shifts is primarily guided by \gls{LOS} paths. In the \gls{mmWave} spectrum, which is the main operational frequency for \gls{LC}-based \gls{RIS}, the \gls{LOS} component is the dominant contributor to received signal strength. Consequently, focusing on these paths is essential for optimizing the system's secrecy performance. The impact of the \gls{nLOS} paths will be investigated in Section~\ref{sec: Performance Comparison}.

\section{LC Phase Shifter Model}
\label{sec: LC phase shifter}
In this section, we characterize the mechanism by which \gls{LC} molecules introduce a controllable phase shift into an impinging electromagnetic wave. We first introduce the baseline phase shifter model at a fixed reference temperature in Section~\ref{LC phase shifter model without temperature impact}. Subsequently, in Section~\ref{subsec: Theory behind the temperature impact on LC}, we develop a model that incorporates the thermodynamic behavior of the \gls{LC} to model the temperature-dependent phase-shift response.

\subsection{LC Phase Shifter Model at a Reference Temperature}
\label{LC phase shifter model without temperature impact}
\gls{LC}-\glspl{RIS} manipulate incoming signals by exploiting the anisotropic electromagnetic properties of \gls{LC} molecules, which can be dynamically reoriented via an applied external electric field ($\vec{E}_{\rm RF}$)~\cite{jimenez2023reconfigurable}. Because \gls{LC} molecules have an elongated, rod-like geometry, their local permittivity depends heavily on whether the electric field vector is aligned with their major or minor axis. Aligning the field with the major axis yields a higher relative permittivity, which consequently maximizes the phase delay induced by the individual \gls{RIS} element. Conversely, alignment with the minor axis yields a lower permittivity. By adjusting the biasing voltage ($V$) across the \gls{LC} layer, the orientation of the molecules can be continuously tuned, allowing for a programmable wireless environment.

The maximum phase tuning range $\Delta\omega_{\max}$ achievable by an \gls{LC} element is fundamentally restricted by the physical dimensions of the cell, the operating frequency, and the maximum material anisotropy, i.e.,
\begin{equation}
    \Delta\omega_{\max} = 2\pi l \left(\sqrt{\varepsilon_{r,\parallel}} - \sqrt{\varepsilon_{r,\perp}}\right) \frac{f}{c},
    \label{eq:omega epsilon}
\end{equation}
where $l$ is the physical length of the phase shifter cell, $f$ is the operating frequency, and $c$ is the speed of light in vacuum. The parameters $\varepsilon_{r,\parallel}$ and $\varepsilon_{r,\perp}$ denote the maximum and minimum relative permittivities, which correspond to configurations where the electric field is perfectly parallel or perpendicular to the molecular alignment vector $\bu_\text{molecule}$, respectively. Both $\varepsilon_{r,\parallel}$ and $\varepsilon_{r,\perp}$ scale continuously with temperature,  which leads to a
temperature-dependent phase-shift response for \gls{LC}-RISs, as will be discussed in detail in Section~\ref{subsec: Theory behind the temperature impact on LC}.

\subsection{Proposed Temperature-Dependent Phase Shifter Model}
\label{subsec: Theory behind the temperature impact on LC} 

\begin{figure}[tbp]
\centering
\begin{subfigure}{0.15\textwidth}
    \centering
    \includegraphics[width=1\textwidth]{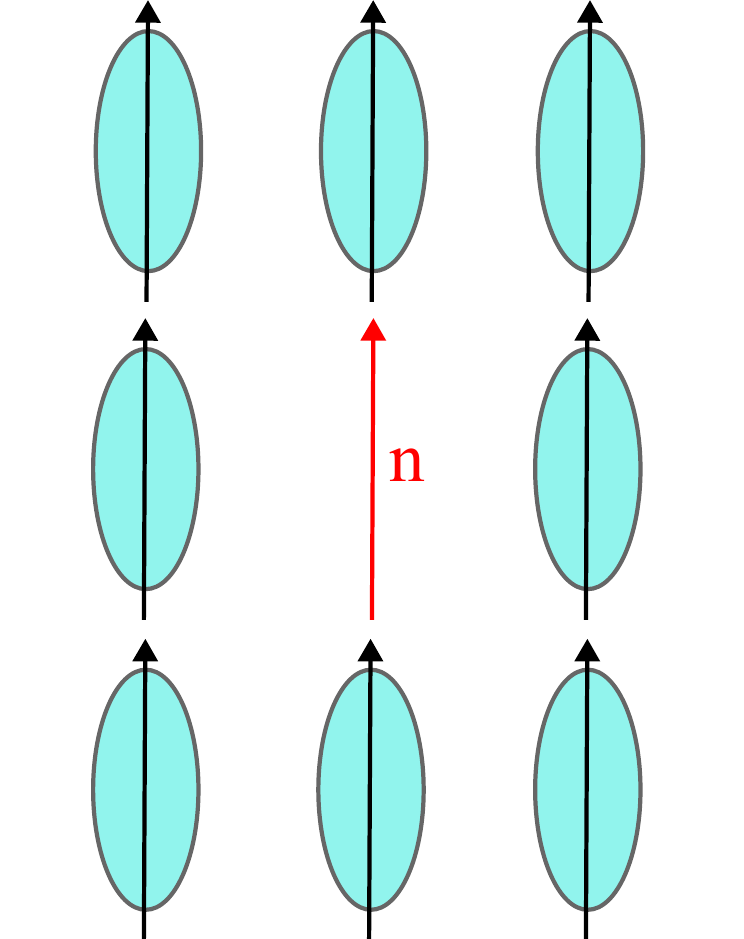}
    \caption{Low temperature}
    \label{fig: alignment temperature a}
\end{subfigure}
\hfill
\begin{subfigure}{0.15\textwidth}
    \centering
    \includegraphics[width=1\textwidth]{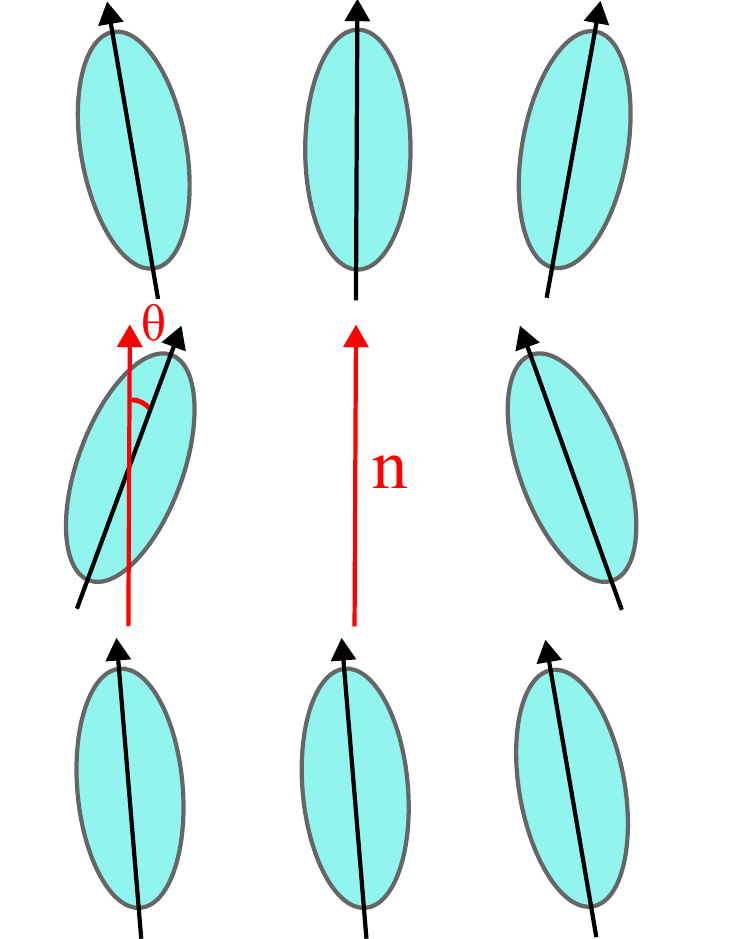}
    \caption{Moderate temperature}
    \label{fig: alignment temperature b}
\end{subfigure}
\hfill
\begin{subfigure}{0.15\textwidth}
    \centering
    \includegraphics[width=1\textwidth]{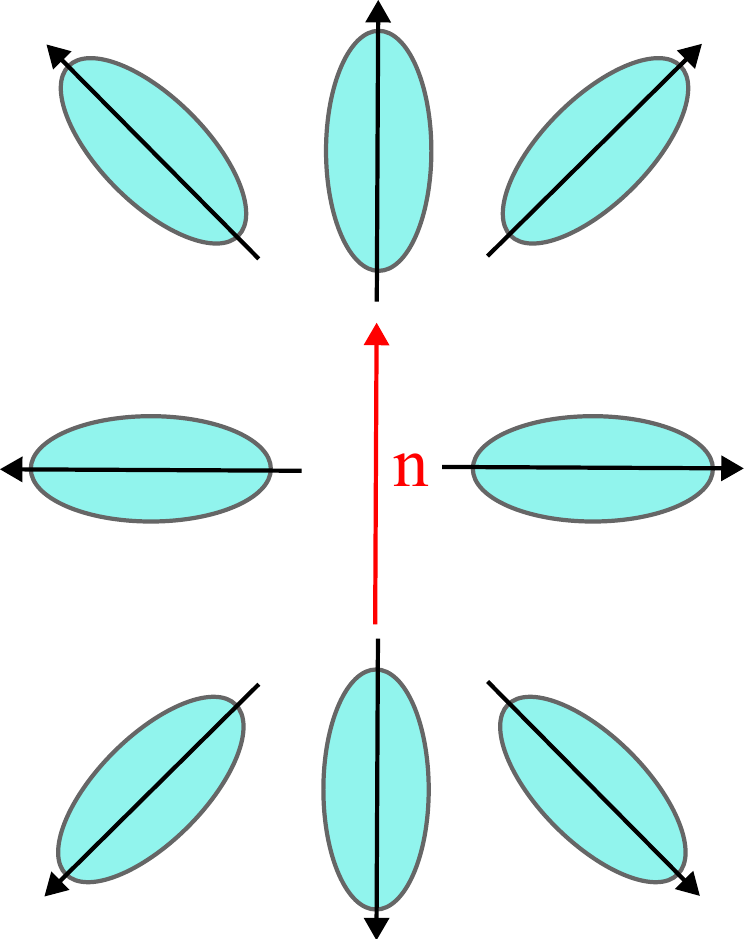}
    \caption{High temperature}
    \label{fig: alignment temperature c}
\end{subfigure}
\caption{Schematic representation of \gls{LC} molecular alignment profiles across different thermal states, where $\bn$ shows the desired direction while individual \gls{LC} molecules may deviate from $\bn$ with angle of $\theta$ due to the temperature.}
\label{fig: alignment temperature}
\vspace{-5mm}
\end{figure}
The \gls{LC} molecules maintain a long-range orientational order along a common vector known as the director, denoted by the unit vector $\bn$ (see Fig.~\ref{fig: alignment temperature}) in the nematic phase. While an applied voltage controls this average molecular orientation to tune the macroscopic permittivity, thermal energy causes the molecules to fluctuate around this mean direction. As temperature increases, these thermal fluctuations disrupt the collective alignment, directly reducing the maximum differential phase shift achievable by the \gls{LC}-RIS in \eqref{eq:omega epsilon}. To quantify this thermal degradation, a scalar order parameter $S(T)$ is defined to capture the macroscopic state of the \gls{LC} fluid. By applying the linear birefringence approximation \cite{Wang2005}, the temperature-dependent refractive index variation can be modeled as:
\begin{equation}
\label{eq: linear define S}
    \sqrt{\varepsilon_{r,\parallel}} - \sqrt{\varepsilon_{r,\perp}} \approx \Delta n_0 S(T),
\end{equation}
where $\Delta n_0$ represents the structural birefringence extrapolated to absolute zero ($0\text{ K}$). By substituting \eqref{eq: linear define S} into \eqref{eq:omega epsilon}, the maximum phase tuning range scales according to:
\begin{subequations}
\label{eq: delta omega max temp}
\begin{align}
    &\Delta\omega_{\max}(T) = 2\pi l \frac{f}{c} \Delta n_0 S(T), \\
    &\Delta\omega_{\max}(T_r) = 2\pi l \frac{f}{c} \Delta n_0 S(T_r) = 2\pi,
\end{align}
\end{subequations}
where $T_r$ denotes a reference temperature at which a full $2\pi$ phase shift range must be achievable.

To evaluate the order parameter $S(T)$, let the orientation of a single \gls{LC} molecule in a 3D Spherical coordinate system be defined by its major axis unit vector $\bu_\text{molecule}$. Within the molecular ensemble, the molecule experiences a mean-field potential that aligns it toward the director $\bn$, which is determined by the applied control voltage. The polar angle between $\bu_\text{molecule}$ and $\bn$ is denoted by $\theta$. Assuming an electric field applied along the $z$-axis prompts the molecules to rotate within this reference frame, the orientation \gls{PDF} across a differential solid angle $\dd\Omega = \sin\theta \dd\theta \dd\phi$ is given by $f(\theta,\phi)$. The scalar order parameter $S$ is defined as the statistical expectation of a function $g(\theta)$ which incorporates the impact of the $\theta$ in the maximum differential phase shift achievable by the \gls{LC}-RIS in \eqref{eq:omega epsilon}~\cite{cheung2002structures}:
\begin{equation}
    \label{eq: S definition}
    S \triangleq \mathbb{E}\{g(\theta)\} = \int_0^{2\pi} \int_0^{\pi} g(\theta) f(\theta,\phi) \sin\theta \, \dd\theta \, \dd\phi.
\end{equation}
To properly characterize nematic order, $g(\theta)$ must satisfy two boundary constraints:
\begin{itemize}
    \item \textbf{Normalization:} For perfect crystalline alignment parallel to the director (Fig.~\ref{fig: alignment temperature a}), $f(\theta,\phi)\sin\theta \to \delta(\theta)$, which must yield $S = 1$. Conversely, for a completely disordered isotropic liquid (Fig.~\ref{fig: alignment temperature c}), $f(\theta,\phi)$ becomes a uniform distribution over the sphere, which must yield $S = 0$.
    \item \textbf{Head-to-Tail Symmetry:} Nematic molecules exhibit inversion symmetry, meaning a molecular orientation at angle $\theta$ is physically indistinguishable from one at $\pi - \theta$.
\end{itemize}
The simplest function satisfying both conditions is the second-order Legendre polynomial~\cite{krzyzanowski2026thermotropic}:
\begin{equation}
    \label{second Legendre polynomial}
    g(\theta) = P_2(\cos\theta) = \frac{1}{2}\left(3\cos^2\theta - 1\right).
\end{equation}
According to statistical mechanics, the structural orientation $f(\theta,\phi)$ is governed by the Boltzmann distribution~\cite{cercignani1988boltzmann}:
\begin{equation}
\label{eq: PDF distribution on T}
    f(\theta,\phi) = \frac{1}{Z} \exp\left(-\frac{U(\theta)}{k_B T}\right),
\end{equation}
where $k_B$ is the Boltzmann constant, $T$ is the absolute temperature, and $Z = \int_{0}^{2\pi}\int_{0}^{\pi} \exp\left(-\frac{U(\theta)}{k_B T}\right)\sin\theta \dd\theta \dd\phi$ is the partition function. Here, $U(\theta)$ represents the mean-field potential energy exerted by neighboring molecules. Under the Maier–Saupe mean-field approximation~\cite{cotter1977consistency}, this potential is directly proportional to the macroscopic order parameter $S$:
\begin{equation}
    \label{eq: energy}
    U(\theta) = -\alpha_i S P_2(\cos\theta),
\end{equation}
where $\alpha_i$ is a material constant representing the intermolecular interaction strength. The negative sign confirms that the system's energy is minimized when the molecules align perfectly with the director ($\theta = 0$). 

Because the microscopic distribution $f(\theta,\phi)$ depends on $U(\theta)$, which in turn depends back on $S$, the system forms a self-consistent field loop. Substituting \eqref{eq: energy} into \eqref{eq: PDF distribution on T} and evaluating the expectation in \eqref{eq: S definition} yields:
\begin{equation}
    \label{eq: S definition 2}
    S = \frac{\int_0^{\pi} P_2(\cos\theta) \exp\left(\frac{\alpha_i S P_2(\cos\theta)}{k_B T}\right) \sin\theta \, \dd\theta}{\int_0^{\pi} \exp\left(\frac{\alpha_i S P_2(\cos\theta)}{k_B T}\right) \sin\theta \, \dd\theta}.
\end{equation}
Equation \eqref{eq: S definition 2} lacks a closed-form analytical solution and must be evaluated numerically for any arbitrary temperature $T$. To
avoid this complexity, Haller proposed an empirical power-law relationship to approximate $S(T)$~\cite{haller1975thermodynamic}:
\begin{equation}
\label{eq: Haller}
    S(T) = \left(1 - \frac{T}{T_c}\right)^\beta,
\end{equation}
where $\beta$ is a material-dependent constant typically ranging between $0.2$ and $0.25$ for nematic mixtures~\cite{Wang2005}, and $T_c$ denotes the clearing temperature threshold at which the material transitions into a fully isotropic liquid ($S=0$).

\begin{figure}[tbp]
    \centering
    \includegraphics[width=0.5\textwidth]{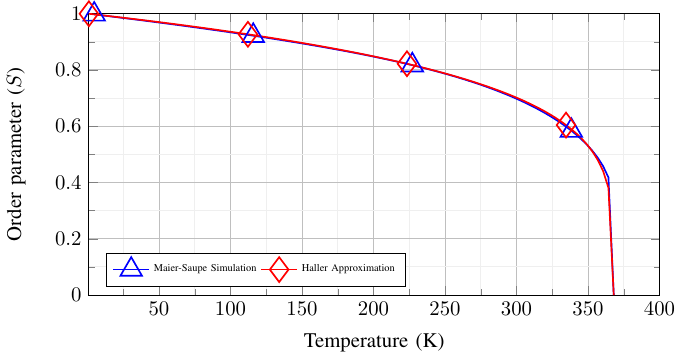}
    \caption{Validation of the nematic order parameter $S(T)$ as a function of temperature. The plot illustrates the numerical solution of the self-consistent Maier-Saupe theory in \eqref{eq: S definition 2} alongside the empirical Haller approximation in \eqref{eq: Haller} for a clearing temperature of $T_c = 368\text{ K}$.}
    \label{fig: Haller}
\end{figure}

\begin{figure}[tbp]
	\centering
	\includegraphics[width=0.5\textwidth]{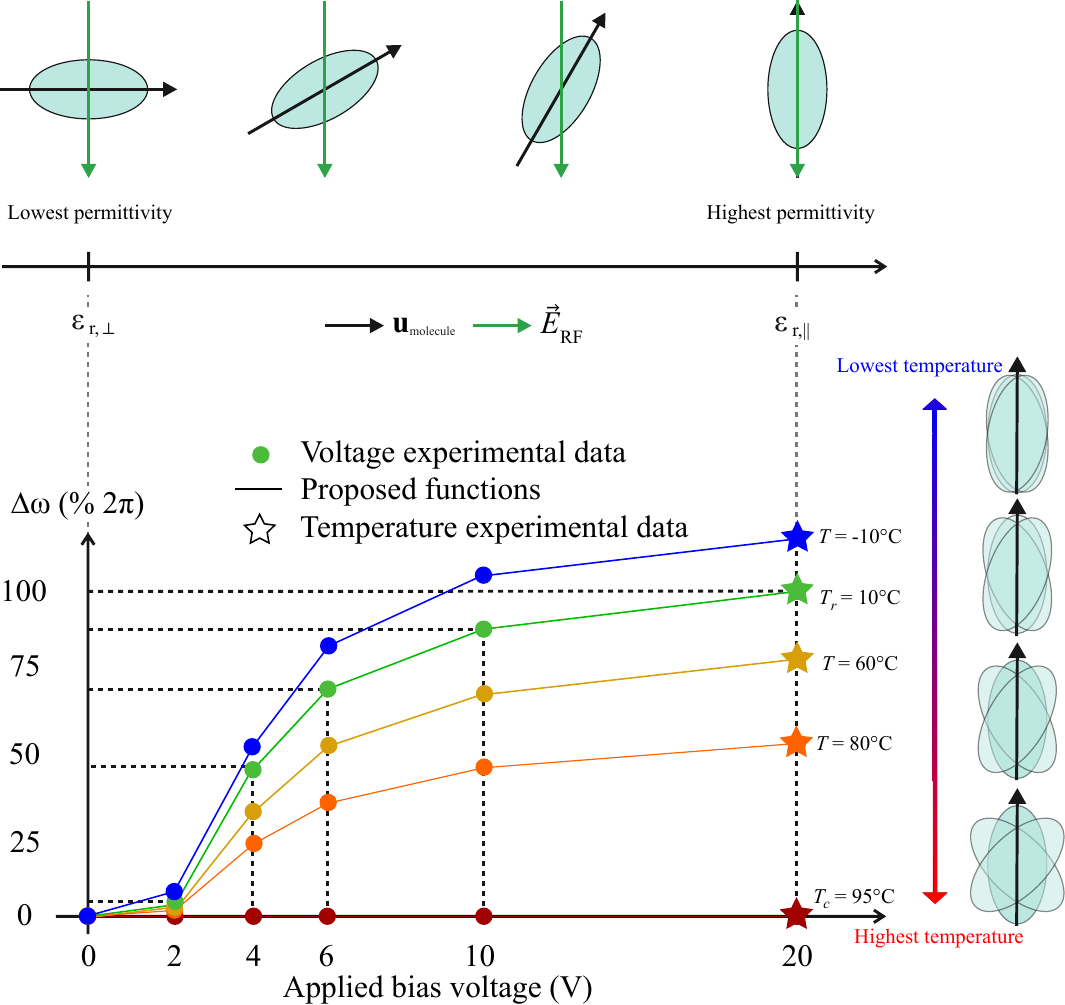}
	\caption{Phase shift versus applied voltage for an \gls{LC} element across different temperatures. The experimental baseline is adapted from \cite{neuder2024architecture,tesmer2021temperature}, with a piece-wise linear function modeled following the approach in \cite{delbari2024fast}.}
	\label{fig:V_phase}
    \vspace{-5mm}
\end{figure}

A comparison between the numerical evaluation of the Maier-Saupe equation \eqref{eq: S definition 2} and the empirical Haller approximation \eqref{eq: Haller} is illustrated in Fig.~\ref{fig: Haller}, demonstrating excellent agreement. By substituting \eqref{eq: Haller} into \eqref{eq: delta omega max temp}, the temperature-dependent maximum phase shift simplifies to
\begin{equation}
    \Delta\omega_{\max}(T) = 2\pi \left(\frac{T_c - T}{T_c - T_r}\right)^\beta.
\end{equation}
Assuming that the baseline minimum phase shift can be calibrated to zero ($\omega_{\min} = 0$) across all thermal states, the complete coupled voltage- and temperature-dependent phase-shift profile $\omega(V,T)$ is formulated as:
\begin{align}
\label{eq: omega in temperature final}
    &\omega(V,T) = \omega(V,T_r) \left(\frac{T_c - T}{T_c - T_r}\right)^\beta, \\
    \label{eq: omega max in temperature}
    &\omega_{\max}(T) \triangleq \omega(V_{\max},T) = 2\pi \left(\frac{T_c - T}{T_c - T_r}\right)^\beta,
\end{align}
where $\omega(V,T_r)$ is the phase shift at the reference temperature (green curve in Fig.~\ref{fig:V_phase})~\cite{neuder2024architecture}. The proposed voltage-to-phase relationship at different temperatures is illustrated in Fig.~\ref{fig:V_phase}.

\section{Temperature-adaptive LC-RIS Phase-shift Design}
\label{sec: Temperature-Aware LC-RIS Phase-shift Design}
In this section, we first formulate an optimization problem that maximizes the secrecy rate for the legitimate user under a given temperature condition. Subsequently, we solve it via an \gls{AO} over \gls{LC}-\gls{RIS} phase shifters and \gls{BS} beamforming.

\subsection{Problem Formulation}
\label{subsec: Problem Formulation}
In this section, we initiate the design of the \gls{RIS} phase shifts by formulating an optimization problem focused on maximizing the secure rate, given specific temperature data\footnote{This data can be provided by equipping thermal sensors on the \gls{LC}-\gls{RIS}.}. Based on the system architecture shown in Fig.~\ref{fig:system model}, the \gls{RIS} is configured to maximize the secure rate defined in \eqref{eq: SNR} for the \gls{MU} under a worst-case scenario. The primary objective is to ensure a maximum secure rate without requiring the precise positioning of the legitimate user or the eavesdropper, rather than providing only their defined spatial zones. Relying solely on the \gls{LOS} links, the problem is formulated as follows:
\begin{subequations}
\label{eq:optimization 1}
\begin{align}
    \text {P1:}\quad&~\underset{\bomega,\bq,\alpha}{\max}~\alpha    \label{eq:optimization 1 a}
    \\&~\text {s.t.} ~~\RS(T)\geq \alpha,\, \forall \bp_u\in\Pset_u,\,\forall\bp_e\in\Pset_e   \label{eq:optimization 1 b}
    \\&\quad\hphantom {\text {s.t.} } 0\leq [\bomega(T)]_n < \omega_\tmax(T), \forall n,  \label{eq:optimization 1 c}
    \\&\quad\hphantom {\text {s.t.} } \|\bq\|_2^2\leq P_t.  \label{eq:optimization 1 d}
\end{align}
\end{subequations}
Here, \eqref{eq:optimization 1 b} is the secure rate constraint for a specific temperature defined in \eqref{eq: SNR}. The constraint \eqref{eq:optimization 1 c} governs the \gls{LC}-\gls{RIS} phase shifts, which are temperature-dependent as detailed in Section \ref{subsec: Theory behind the temperature impact on LC}, and $\omega_\tmax(T)$ was defined in \eqref{eq: omega max in temperature}, while \eqref{eq:optimization 1 d} enforces the \gls{BS} transmit power limits. Here, we aim to maximize the variable $\alpha$ denoting the worst-case secure rate.

Problem P1 is inherently non-convex, primarily due to the mathematical structure of constraint \eqref{eq:optimization 1 b}. Furthermore, the coupling of the vector variables $\bq$ and $\bomega$ within this constraint complicates the search for a global optimum. To address these challenges, we decompose the problem into two sub-problems and utilize \gls{AO} to iteratively maximize the objective function at each stage.

\subsection{RIS Phase-shift Design}
\label{sec: RIS Phase-shift Design}
Assuming a fixed beamformer, we first maximize the secrecy rate by optimizing the \gls{RIS} phase shifts. The phase shift configuration subproblem is given by:
\begin{subequations}
\label{eq:optimization 3}
\begin{align}
    \text {P2:}\quad&~\underset{\bomega,\alpha}{\max}~\alpha
    \\\label{eq:optimization 3 b}
    &~\text {s.t.} ~~\RS(T)\geq \alpha,\, \forall \bp_u\in\Pset_u,\,\forall\bp_e\in\Pset_e
    \\
    &\quad\hphantom {\text {s.t.} } 0\leq [\bomega(T)]_n < \omega_\tmax(T), \forall n.
\end{align}
\end{subequations}
Problem P2 is inherently non-convex, primarily due to the non-convex nature of the phase-shift constraint \eqref{eq:optimization 3 b} \gls{w.r.t.} $[\bomega]_n,\,\forall n$. To facilitate a more tractable formulation without loss of generality, we maximize $\widetilde{\RS}(T)$ as defined in \eqref{eq: SNR b} rather than the original $\RS(T)$ in \eqref{eq: SNR a}. Note that omitting the $[\cdot]^+$ operator does not alter the optimization results; if the optimal solution yields $\widetilde{\RS}(T)>0$, the operator is redundant, whereas an $\widetilde{\RS}(T)<0$ simply indicates that the achievable secrecy rate is zero. By introducing an auxiliary variable $\gamma$ such that $\log(\gamma) \triangleq \alpha$, the constraint \eqref{eq:optimization 3 b} can be reformulated as:
\begin{equation}
\label{eq: gamma definition}
    \frac{1+\SNR_u(T)}{1+\SNR_e(T)}\geq\gamma\Rightarrow (\SNR_u(T)-\gamma\SNR_e(T))\geq\gamma-1.
\end{equation}
 Here, we can also decompose each $\SNR_u$ and $\SNR_e$ in terms of a vector including exponential of \gls{RIS} phase shifts
 \begin{equation}
     \label{eq: phase shift vector}
     \bs(T)\defeq[\e^{\jj[\bomega(T)]_1}, \cdots, \e^{\jj[\bomega(T)]_N}]^\Trans.
 \end{equation}
 With this assumption, we have:
\begin{subequations}
    \label{eq: SNR in term of s}
    \begin{align}
        \SNR_u(T)=&\bs^\Herm(T)\bA_u\bs(T),\\
        \SNR_e(T)=&\bs^\Herm(T)\bA_e\bs(T),
    \end{align}
\end{subequations}
where $\bA_g=\frac{\diag(\bh_{r,g}^\Herm)\bH_t\bq\bq^\Herm\bH_t^\Herm\diag(\bh_{r,g})}{\sigma_n^2},\,g=\{u,e\}$. 
To solve the formulated Problem
P2, we introduce two distinct approaches, each presenting unique trade-offs. The first is an \gls{SDP}-based method, delivering excellent accuracy at the expense of high computational complexity. This can be considered as an upper-bound achievable secrecy rate. The second is a low-complexity method; although its performance is marginally lower than that of the \gls{SDP}-based approach, it is highly scalable and particularly well-suited for extremely large \gls{RIS}s.

\subsubsection{SDP-based method}
\label{sec: Optimization-based method}
To tackle the non-convexity of Problem P2 in the first method, we transform the problem into an \gls{SDP}-based problem. Let us define $\bS(T)\defeq\bs(T)\bs^\Herm(T)$, and $\bA_{u,e}(\gamma)\defeq\bA_u(\bp_u)-\gamma\bA_e(\bp_e)$ where $\bA_{u,e}(\gamma)$ is a function of $\bp_u,\,\bp_e,$ and $\gamma$ but we dropped $\bp_u,\,\bp_e,$ for notational simplicity. In addition, we omit the explicit temperature dependence $(T)$ and denote the phase-shift matrix simply as $\bS$ in the subsequent derivations. After applying these reformulations in P2, and because the logarithm function is increasing monotonically, problem P2 can be changed
to P3 in the following:
\begin{align}
\label{eq:optimization 4}
    \text {P3:}&~\underset{\bS,\gamma}{\max}~\gamma
    \\&~\text{s.t.}~~\text{C1: }\tr(\bA_{u,e}(\gamma)\bS)\geq\!\gamma\!-\!1, \forall (\bp_u,\bp_e)\!\in\Pset_u\!\times\!\Pset_e,\nonumber
    \\&\quad\hphantom {\text {s.t.} } \text{C2: } 0\leq [\bomega(T)]_n < \omega_\tmax(T), \forall n,\nonumber
    \\&\quad\hphantom {\text {s.t.} }\text{C3: } \bS\succeq 0,\text{C4: } \rank(\bS)=1, \text{C5: }\diag(\bS)=\bone_N.\nonumber
\end{align}
Despite the previous transformations, Problem~P3 still remains non-convex due to the non-convex nature of constraints C2 and C4 \gls{w.r.t.} $\bS$, as well as the coupling between the auxiliary variable $\gamma$ and the phase-shift matrix $\bS$ in C1. In the following, we decouple these variables and resolve the non-convexities associated with each constraint.
\paragraph{Rank constraint C4} To address this issue, we adopt the exploited penalty method in \cite{Yu2020power,ghanem2022optimization,delbari2026far}.
The basic idea is to replace the rank constraint with the inequality $\|\bS\|_*-\|\bS\|_2 \leq 0$, which holds only if $\bS$ has rank smaller than or equal to one. While the new constraint is still non-convex, one can apply the first-order Taylor approximation to make it convex. Let $\bS^{(i)}$ denote the value of matrix $\bS$ in the $i$th iteration. Based on the first-order Taylor approximation:
\begin{equation}
\label{eq: taylor approximation}
    \|\bS\|_2\geq\|\bS^{(i)}\|_2+\tr\big(\blambda_{\max}(\bS^{(i)})\blambda_{\max}^\Herm(\bS^{(i)})(\bS-\bS^{(i)})\big),
\end{equation}
By exploiting the penalty method \cite{Yu2020robust} and applying \eqref{eq: taylor approximation} into the cost function of P3, we have
\begin{subequations}
\label{eq:optimization 5}
\begin{align}
    \text {P4:}\quad&~\underset{\bS,\gamma}{\max}~\gamma-\eta^{(i)}\Big(\|\bS\|_*-\|\bS^{(i)}\|_2-\tr\big(\blambda_{\max}(\bS^{(i)})\nonumber
    \\&\quad\quad\quad\times\blambda_{\max}^\Herm(\bS^{(i)})(\bS-\bS^{(i)})\big)\Big)
    \\&~\text {s.t.} ~~ \text{C1, C2, C3, C5}.
\end{align}
\end{subequations}
Here, $\eta^{(i)}$ is the penalty factor at iteration $i$, which increases gradually. By choosing a sufficiently large $\eta$, problems P3 and P4 become equivalent. We will tackle the non-convexity of C2 \gls{w.r.t.} $\bS$ in the following.
\paragraph{Constraint C2: $0\leq [\bomega(T)]_n < \omega_{\max}(T),\,\forall n$} Although the constraint C2 is linear in $[\bomega(T)]_n, \forall n$, it is highly non-convex in the new defined variable $\bS$. To address this issue, we extract features of $\bS$ that are informative about $\omega_{\max}(T)$ and can be used to enforce C2. Note that, based on \eqref{eq: omega max in temperature}, $\omega_\tmax$ exceeds $2\pi$ when $T < T_r$. This allows the \gls{LC}-\gls{RIS} phase shifts to map directly into the $0-2\pi$ range. Conversely, if \(T > T_r\), then $\omega_\tmax$ is less than \(2\pi\). For this latter case, we propose the following solution. We first present two lemmas, and based on them, we reformulate this constraint to another constraint in terms of $\bS$ satisfying C2. 
\begin{lem}
\label{lemma integral}
    Let us assume $[\bomega]_n$ is empirically distributed uniform in interval $[0,\omega_{\max}]$. For sufficiently large $N$, matrix $\bS$ that satisfies the constraints C2-C5, also satisfies:
    \begin{equation}
    \sum_{i=1}^N [\bS]_{n,i}=\frac{\e^{\jj[\bomega]_n}N(1-\e^{-\jj\omega_{\max}})}{{\jj\omega_{\max}}},\forall n.
    \end{equation}
\end{lem}
\begin{proof}
    Recalling $\bS=\bs\bs^\Herm$ with $\bs=[\e^{\jj[\bomega]_1},\,\cdots,\,\e^{\jj[\bomega]_N}]$, we have
    \begin{align}
        \sum_{i=1}^N [\bS]_{n,i}&=\e^{\jj[\bomega]_n}\sum_{i=1}^N \e^{-\jj[\bomega]_i}=\e^{\jj[\bomega]_n}N\sum_{i=1}^N\frac{\e^{-\jj[\bomega]_i}}{N}\\
        &\stackrel{(a)}{=}\frac{\e^{\jj[\bomega]_n}N}{\omega_{\max}} \int_{0}^{\omega_{\max}}\e^{-\jj\omega} \dd\omega=\frac{\e^{\jj[\bomega]_n}N(1-\e^{-\jj\omega_{\max}})}{\jj\omega_{\max}},
    \end{align}
    where $(a)$ holds based on law of large numbers \cite{papoulis2002probability}.
\end{proof}

\begin{lem}
\label{lemma constraint C2}
    The constraint $0\leq[\bomega]_n\leq\omega_\tmax,$
    , where $\pi<\omega_{\max}<2\pi$, is equivalent to
     $\real([\bs]_n)+\tan(\frac{\omega_{\max}}{2})\imag([\bs]_n)\leq1$,
     where $[\bs]_n=\e^{\jj[\bomega]_n}$ and $\real(\cdot)$ and $\imag(\cdot)$ denote the real and imaginary parts, respectively.
\end{lem}
\begin{proof}
    Starting by $[\bs]_n=\e^{\jj[\bomega]_n}=\cos([\bomega]_n)+\jj\sin([\bomega]_n)$, we can express $\real([\bs]_n)+\tan(\frac{\omega_{\max}}{2})\imag([\bs]_n)$ as $\cos([\bomega]_n)+\tan(\frac{\omega_{\max}}{2})\sin([\bomega]_n)$. By substituting this into the inequality, it yields that
    \begin{align}
        \tan(\frac{\omega_{\max}}{2})\sin([\bomega]_n)\leq1-\cos([\bomega]_n).
    \end{align}
    Using the known trigonometric equations $\sin([\bomega]_n)=2\sin(\frac{[\bomega]_n}{2})\cos(\frac{[\bomega]_n}{2})$ and $1-\cos([\bomega]_n)=2\sin^2(\frac{[\bomega]_n}{2})$, we can transform the inequality into:
    \begin{align}
        2\tan(\frac{\omega_{\max}}{2})\sin(\frac{[\bomega]_n}{2})\cos(\frac{[\bomega]_n}{2})\leq2\sin^2(\frac{[\bomega]_n}{2}).
    \end{align}
    We can divide out $\sin(\frac{[\bomega]_n}{2})$ from both sides of the inequality since $\sin(\frac{[\bomega]_n}{2})\geq0$ for $0\leq[\bomega]_n\leq2\pi$. Consequently, the inequality simplifies to
    \begin{align}
        \tan(\frac{\omega_{\max}}{2})\cos(\frac{[\bomega]_n}{2})\leq\sin(\frac{[\bomega]_n}{2}).
    \end{align}
    This is equivalent to
    \begin{align}
       \begin{cases}
            \textbf{Case 1:} \tan(\frac{\omega_{\max}}{2})\leq\tan(\frac{[\bomega]_n}{2}),\,\,0<[\bomega]_n<\pi,\\
            \textbf{Case 2:} \tan(\frac{\omega_{\max}}{2})\geq\tan(\frac{[\bomega]_n}{2}),\,\,\pi<[\bomega]_n<2\pi.
        \end{cases}
    \end{align}
    Case 1 always holds because $\tan(\frac{\omega_{\max}}{2})<0$ when $\pi<\omega_{\max}<2\pi$ while $\tan(\frac{[\bomega]_n}{2})>0$. Case 2 is satisfied as long as $[\bomega]_n\leq\omega_{\max}$, thus concluding the proof.
\end{proof}
According to Lemma~\ref{lemma integral} and Lemma~\ref{lemma constraint C2}, we can conclude that the constraint C2 can be written as
\begin{equation}
    \widehat{\text{C2}}:\real(\zeta\sum_{i=1}^N [\bS]_{n,i})+\tan(\frac{\omega_{\max}}{2})\imag(\zeta\sum_{i=1}^N [\bS]_{n,i})\leq1, \forall n,
\end{equation}
where $\zeta\defeq\frac{\jj\omega_{\max}}{N(1-\e^{-\jj\omega_{\max}})}$. Unlike C2, constraint $\widehat{\text{C2}}$ is convex in $\bS$. These two constraints are equivalent under three key assumptions: $(i)$ a sufficiently large number of \gls{RIS} elements ($N$), $(ii)$ an empirically uniform distribution of phase shifts in the interval $[0,\,\omega_{\max}]$, and $(iii)$ $\pi<\omega_{\max}<2\pi$. Based on our observations, $N\geq50$ is sufficient for this approximation. While a uniform distribution is not generally guaranteed, the individual phase shifts tend to be rather random in \gls{NF} regime when \gls{RIS} serves an area, making the assumption reasonable. Regarding the third assumption, experimental results confirm that even at higher temperatures typical in outdoor scenarios, $\omega_{\max}$ does not drop below $\pi$ \cite{tesmer2021temperature}.

Therefore, problem P4 is transformed into the following optimization problem.
\begin{subequations}
\label{eq:optimization 6}
\begin{align}
    \text {P5:}\quad&~\underset{\bS,\gamma}{\max}~\gamma-\eta^{(i)}\Big(\|\bS\|_*-\|\bS^{(i)}\|_2-\tr\big(\blambda_{\max}(\bS^{(i)})\nonumber
    \\&\quad\quad\quad\times\blambda_{\max}^\Herm(\bS^{(i)})(\bS-\bS^{(i)})\big)\Big)
    \\&~\text {s.t.} ~~\text{C1}, \widehat{\text{C2}},\text{ C3, C5}.
\end{align}
\end{subequations}

\paragraph{Coupled $\gamma$ and $\bS$ in C1} To address the coupling of $\bS$ and $\gamma$ in C1, we employ \gls{AO}, where one variable is fixed while the other is optimized. On one hand, when $\gamma$ is fixed, the optimization problem P5 becomes convex in terms of the matrix $\bS$ because the objective function is concave, and the constraints define a convex set. Therefore, it can be efficiently solved using standard convex optimization solvers such as CVX \cite{cvx}. On the other hand, when the matrix $\bS$ is fixed, the problem P5 is linear
in terms of $\gamma$ and its closed-form solution is given by:
\begin{equation}
\label{eq: best gamma}
    \gamma=\underset{\forall \bp_u\in\Pset_u,\,\forall\bp_e\in\Pset_e}{\min}~\frac{\tr(\bA_u(\bp_u)\bS)+1}{\tr(\bA_e(\bp_e)\bS)+1}.
\end{equation}
\begin{algorithm}[t]
\caption{Proposed SDP-based Algorithm for \gls{LC}-\gls{RIS} phase shift design}\label{alg:cap}
\begin{algorithmic}[1]
\STATE \textbf{Initialize:} $\omega_{\max}$, $\bs^{(0)}=\e^{\jj\omega_{\max}\times\mathrm{rand}(N)},\bS^{(0)}=\bs^{(0)}{\bs^{(0)}}^\Herm$.
\WHILE{$|\log_2(\gamma^{(j)})-\log_2(\gamma^{(j-1)})|\geq\epsilon_2$ and $j\leq J_{\max}$}
    \WHILE{$\|\bS^{(i)}-\bS^{(i-1)}\|_F^2\geq\epsilon_1$ and $i\leq I_{\tmax}$}
    \STATE Solve convex P5 for given $\bS^{(i-1)}$ and $\gamma$, and store the intermediate solution $\bS$.
    \STATE Set $i = i + 1$ and update $\bS^{(i)}=\bS$ and $\eta^{(i)} =5\eta^{(i-1)}$.
    \ENDWHILE
    \STATE Calculate new $\gamma$ according to the \eqref{eq: best gamma} and set $j = j + 1$.
\ENDWHILE
\end{algorithmic}
\end{algorithm} 
The proposed \gls{SDP}-based algorithm is summarized in Algorithm \ref{alg:cap}. It consists of two loops; the inner loop finds a
rank-one solution to Problem P5, and the outer loop maximizes the secure rate.

\subsubsection{Low-complexity method}
\label{sec: Analytical method}

\begin{algorithm}[t]
\caption{Proposed low-complexity Algorithm for the \gls{LC}-\gls{RIS} phase shift design}
\label{alg:cap 2}
\begin{algorithmic}[1] 
\STATE \textbf{Input:} Location sets $\Pset_u$ and $\Pset_e$, $T$, $\omega_\tmax(T)$, $\mu_0 > 0$, $J_\text{max}$, and $I_\text{max}$.
\STATE  \textbf{Initialize:} $\gamma$, $\tilde{\bs}(T)$, and $\RS_\text{final} = -\infty$.
\FOR{$j = 1$ \TO $J_\text{max}$}
\STATE Reset smoothing parameter $\mu = \mu_0$.
    \FOR{$i = 1$ \TO $I_\text{max}$}
    \STATE Update smoothing parameter $\mu \leftarrow \mu \times 0.88$.
        \STATE Compute $\bA_{u,e}(\gamma)$ and $\bPhi_{u,e}(\gamma),\,\forall(\bp_u,\bp_e)\in\Pset_u\times\Pset_e$ using \eqref{eq: calculation matrix A} and \eqref{eq: calculation matrix Phi}.
        \STATE Calculate the surrogate vector $\bbeta_{u,e}(\gamma),\,\forall(\bp_u,\bp_e)\in\Pset_u\times\Pset_e$ by \eqref{eq: beta calculation}.
        \STATE Evaluate the secrecy rate $\widetilde{\RS}(\bp_u,\bp_e),\, \forall (\bp_u,\bp_e) \in \Pset_u \times \Pset_e$ using $\tilde{\bs}(T)$.
        \STATE Calculate the spatial weights $w(\bp_u,\bp_e)$ and $w_\text{total}$ via the LSE approximation in \eqref{eq: LSE lower bound}.
        \STATE Compute the unconstrained optimal phase angles $\bomega(T)$ using \eqref{eq:LSE}.
        \STATE Project $\bomega(T)$ into the feasible temperature-constrained domain to obtain $\bs(T)$ via the wrapping function in \eqref{eq: phase thresholding}.
        \IF{$\underset{\bp_u,\bp_e}{\min}\widetilde{\RS}(\bp_u,\bp_e) > \RS_\text{final}$}
            \STATE $\RS_\text{final} \leftarrow \underset{\bp_u,\bp_e}{\min}\widetilde{\RS}(\bp_u,\bp_e)$,\quad $\bs^\star(T) \leftarrow \bs(T)$
        \ENDIF
        
        \STATE Update $\tilde{\bs}(T) \leftarrow \bs(T)$.
    \ENDFOR
    \STATE Update the auxiliary variable $\gamma$ using \eqref{eq: best gamma}.
\ENDFOR
\STATE \textbf{Output:} The optimal temperature-adaptive phase-shift vector $\bs^\star(T) = \bs(T)$.
\end{algorithmic}
\end{algorithm}

Although the \gls{SDP}-based algorithm detailed in Section~\ref{sec: Optimization-based method} provides a high-performance benchmark, its substantial computational overhead (analyzed in Section~\ref{sec: Algorithm and complexity analysis}) limits its practical applicability for extremely large \gls{LC}-\glspl{RIS}. Consequently, this section develops a highly scalable, low-complexity alternative. To address the non-convexity of Problem~P2, with the help of the definitions in \eqref{eq: gamma definition}, \eqref{eq: phase shift vector}, and \eqref{eq: SNR in term of s}, we reformulate the optimization problem as follows:
 \begin{subequations}
\label{eq:optimization 7}
\begin{align}
    \text{P6:}&~\underset{\gamma,\bs(T)}{\max}~\gamma
    \\&~\text {s.t.} ~~\text{C1: }\frac{1+\SNR_u(T)}{1+\SNR_e(T)}\geq\gamma, \forall (\bp_u,\bp_e)\!\in\!\Pset_u\!\times\!\Pset_e,
    \\&\quad\hphantom {\text {s.t.} }\text{C2: } 0\leq\arg\left(\bs\left(T\right)\right)\leq\omega_\tmax(T),
    \\&\quad\hphantom {\text {s.t.} }\text{C3: } |[\bs(T)]_n|=1,\,\forall n.
\end{align}
\end{subequations}
We propose to optimize the auxiliary variable $\gamma$ and the phase-shift vector $\bs(T)$ in an alternating manner. However, even for a fixed $\gamma$, Problem~P6 remains non-convex due to the unit-modulus constraint C3. To render the optimization tractable, we maximize a surrogate lower bound \cite{Shen2019}. Let us define
\begin{align}
\label{eq: calculation matrix A}
    \bA_{u,e}(\gamma)&\defeq\bA_u(\bp_u)-\gamma\bA_e(\bp_e),\\
    \bPhi_{u,e}(\gamma)&\defeq\bA_{u,e}(\gamma)-\lambda_{\min}(\bA_{u,e}(\gamma))\bI_N,
    \label{eq: calculation matrix Phi}
\end{align}
where $\lambda_{\min}(\cdot)$ extracts the minimum eigenvalue of a matrix. Due to the subtracted term $-\gamma\bA_e(\bp_e)$ in \eqref{eq: calculation matrix A}, this minimum eigenvalue is inherently negative. Utilizing these definitions, we can substitute the quadratic term in constraint C1, $\bs^\Herm(T)\bA_{u,e}(\gamma)\bs(T)-\gamma+1\geq0$, with a surrogate lower bound provided by the following lemma.
\begin{lem}
\label{lem: lower bound}
    A valid lower bound for $\bs(T)^\Herm\bA_{u,e}(\gamma)\bs(T)$ is
    \begin{equation}
        \lambda_{\min}(\bA_{u,e}(\gamma)) N+2\real\{\bs^\Herm(T)\bbeta_{u,e}(\gamma)\}-\tilde{\bs}^\Herm(T)\bPhi_{u,e}(\gamma)\tilde{\bs}(T),
    \end{equation}
    where
    \begin{equation}
        \label{eq: beta calculation}
        \bbeta_{u,e}(\gamma)\defeq\bPhi_{u,e}(\gamma)\tilde{\bs}(T).
    \end{equation}
 Here, $\gamma$ and the constant vector $\tilde{\bs}(T)$ are fixed, and equality is achievable by $\bs(T)=\tilde{\bs}(T)$.
\end{lem}
\begin{proof}
The proof is similar to that in \cite[Lemma 3]{delbari2026wideband}.
\end{proof}
By leveraging Lemma~\ref{lem: lower bound}, we can transform Problem~P6 into the maximization of its surrogate objective. To facilitate a closed-form update, we temporarily relax the temperature-dependent phase-range constraint C2, yielding the following relaxed subproblem:
\begin{subequations}
\label{eq:optimization 8}
\begin{align}
    \text{P7:}&~\underset{\bs(T)}{\max}~\underset{\bp_u,\bp_e}{\min}~\real\{\bs^\Herm(T)\bbeta_{u,e}(\gamma)\}
    \\&~\text {s.t.} ~~\text{C3}
\end{align}
\end{subequations}
Evidently, the relaxed Problem~P7 is not strictly equivalent to the original Problem~P6 and must be solved iteratively. A naive optimization strategy would be to fix the previous state $\tilde{\bs}(T)$ written in \eqref{eq: beta calculation} at each iteration and update $\bs(T)$ focusing exclusively on the specific spatial location that yields the lowest secrecy rate. However, this hard-minimum approach is highly vulnerable to the ping-pong effect, where the algorithm continuously oscillates between a finite set of active location constraints without converging. To circumvent this instability, rather than directly maximizing the non-smooth, worst-case \gls{RS} in Problem~P7, we maximize its \gls{LSE} approximation. This substitution provides a smooth, analytically tractable surrogate objective \cite[Lemma~4]{delbari2026wideband}.
\begin{lem}
    \label{lem: LSE}
    A valid lower bound for $\underset{n}{\min}~x_n$ is given by
    \begin{equation}
        \label{eq: lem LSE}
        \underset{n}{\min}~x_n \geq -\mu\log\left(\sum_n \exp\left(-\frac{x_n}{\mu}\right)\right),
    \end{equation}
    where $\mu > 0$ is a smoothing parameter. This approximation becomes exact (i.e., equality is achieved) as $\mu\to0$.
\end{lem}
\begin{proof}
    The proof follows similar steps to those detailed in \cite[Lemma~4]{delbari2026wideband}.
\end{proof}
Applying Lemma~\ref{lem: LSE}, we obtain:
\begin{subequations}
\label{eq: LSE lower bound}
    \begin{align}
    &\underset{\bp_u,\bp_e}{\min}\widetilde{\RS}(\bp_u,\bp_e) \geq -\mu\log(w_\text{total}),\\
&w_\text{total}\defeq\sum_{\bp_u\in\Pset_U}\sum_{\bp_e\in\Pset_e}w(\bp_u,\bp_e),\\
        &w(\bp_u,\bp_e)\defeq\exp(-\frac{\widetilde{\RS}(\bp_u,\bp_e)}{\mu}),
    \end{align}
\end{subequations}
where the argument $\widetilde{\RS}(\bp_u,\bp_e)$ is computed using the phase-shift vector $\tilde{\bs}(T)$ obtained from the preceding iteration. By assigning exponentially larger weights to the spatial scenarios exhibiting the lowest secrecy rates, the unconstrained optimal phase-shift vector can be extracted directly via its phase angle:
\begin{equation}
\label{eq:LSE}
\bomega(T)= \arg\!\left(\sum_{\bp_u\in\Pset_U}\sum_{\bp_e\in\Pset_e}\!\!\!\frac{w(\bp_u,\bp_e)}{w_\text{total}}\bbeta_{u,e}(\gamma)\!\right)\!\!.
\end{equation}
While this analytically derived vector is optimal for the relaxed Problem~P7, it does not necessarily satisfy the temperature-dependent phase-range constraint C2 of the original Problem~P6. To restore feasibility, we project the solution back into the valid phase domain using the following piecewise wrapping function:
\begin{equation}
    \label{eq: phase thresholding}
    [\bs(T)]_n=\begin{cases}
    &\e^{\jj[\bomega(T)]_n},\quad \text{If } 0\leq[\bomega]_n\leq\omega_\tmax\\
    &\e^{\jj\omega_\tmax},\quad \text{If } \omega_\tmax\leq[\bomega]_n\leq\frac{2\pi+\omega_\tmax}{2}\\
    &1,\quad \text{Otherwise.}
    \end{cases}
\end{equation}
The vector $\bs(T)$ obtained via \eqref{eq: phase thresholding} constitutes a strictly feasible solution for Problem~P6. This procedure is executed iteratively, with the reference vector $\tilde{\bs}(T)$ updated by the newly projected $\bs(T)$ at each step. Once the sequence of phase-shift vectors converges, the auxiliary variable $\gamma$ is updated according to \eqref{eq: best gamma}. The complete procedure for this low-complexity approach is summarized in Algorithm~\ref{alg:cap 2}.

\subsection{Beamformer Design}
\label{sec: Beamformer Design}
In this step, we assume $\bomega$ is fixed and the only variable of the problem is $\bq$. In addition, since the LOS link is the dominant path at higher frequencies, we design the beamformer based on the \gls{LOS} link. This assumption is often valid, especially because both the BS and \gls{RIS} are positioned at elevated locations above the ground. Based on this assumption, $\bH_t^\LOS$ can be decomposed as:
\begin{equation}
\label{eq: H_t}
    \bH_t^\LOS=c_0\ba_\text{RIS}(\bp_\BS)\ba^\Herm_\BS(\bp_\RIS),
\end{equation}
where $\ba_\RIS(\cdot)$ and $\ba^\Herm_\BS(\cdot)$ are the steering vectors at the \gls{RIS} and \gls{BS}, respectively, and $\|\ba_\RIS(\cdot)\|=\|\ba^\Herm_\BS(\cdot)\|=1$ where their elements are defined in \eqref{Eq:LoSnear}. Due to the large \gls{LC}-\gls{RIS}, we assume \gls{NF} model for $\ba_\RIS(\cdot)$ \cite{delbari2024nearfield}. Moreover, $\bp_\BS$ and $\bp_\RIS$ are the locations of \gls{BS} and \gls{RIS}, respectively. $c_0$ denotes the channel attenuation factor of the \gls{LOS} link. By assuming a fixed $\bomega$, the problem P1 reduces to the following sub-problem:
\begin{subequations}
\label{eq:optimization 2}
\begin{align}
    \text {P8:}\quad&~\underset{\bq,\alpha}{\max}~\alpha
    \\&~\text {s.t.} ~~\RS(T)\geq \alpha,\, \forall \bp_u\in\Pset_u,\,\forall\bp_e\in\Pset_e,
    \\&\quad\hphantom {\text {s.t.} } \|\bq\|_2^2\leq P_t.
\end{align}
\end{subequations}
The optimal beamformer is obtained using the following lemma \cite{delbari2024temperature}.
\begin{lem}
\label{lemma beamforming}
Under the assumption of blocked direct links for both the legitimate user and the eavesdropper, a dominant \gls{LOS} channel, and given fixed \gls{RIS} phase shifts, $\bq=\sqrt{P_t}\ba_\BS(\bp_\RIS)$ represents the optimal beamformer for P8.
\end{lem}
\begin{proof}
By inserting the \gls{LOS} channel model $\bH_t^\LOS$ from \eqref{eq: H_t} into the secrecy rate expression \eqref{eq: SNR} and omitting the $\bh_{d,g},\,g=\{u,e\}$ due to blockage, the secrecy rate simplifies to:
\begin{equation}
    \RS=\Big[\log\big(\frac{1+\overbrace{|\bh_{r,u}\bGamma c_0\ba_\RIS(\bp_\BS)|^2/\sigma^2_n}^{\zeta_u(\bp_u)}\times   \overbrace{|\ba^\Herm_\BS(\bp_\RIS)\bq|^2}^{\mu}}{1+\underbrace{|\bh_{r,e}\bGamma c_0\ba_\RIS(\bp_\BS)|^2/\sigma^2_n}_{\zeta_e(\bp_e)}\times   \underbrace{|\ba^\Herm_\BS(\bp_\RIS)\bq|^2}_{\mu}}\big)\Big]^+,
\end{equation}
where $\zeta_u(\bp_u)$ and $\zeta_e(\bp_e)$ are constant scaling factors in terms of the given \gls{RIS} phase shifts, whereas $\mu\in\Rset$ represents a design variable governed by $\bq$. In the regime where $\zeta_u(\bp_u)>\zeta_e(\bp_e),\,\forall \bp_u\in\Pset_u, \forall \bp_e\in\Pset_e$ holds, $\RS$ grows monotonically in $\mu$. Applying the Cauchy-Schwarz inequality, the upper bound of $\mu$ is $P_t$, which is achieved by aligning the beamformer as $\bq=\sqrt{P_t}\ba_{\BS}(\bp_{\RIS})$ \cite{tse2005fundamentals}. Conversely, if there exists at least one pair $(\bp_u,\bp_e)$, $\bp_u\in\Pset_u$ and $\bp_e\in\Pset_e$, such that $\zeta_u(\bp_u)\leq\zeta_e(\bp_e)$, the secrecy rate $\RS$ drops to zero regardless of the choice of $\mu$. Thus, selecting $\bq=\sqrt{P_t}\ba_{\BS}(\bp_{\RIS})$ remains optimal under all conditions, which concludes the proof.
\end{proof}

\subsection{Algorithm and complexity analysis}
\label{sec: Algorithm and complexity analysis}
The complexity of Algorithm~\ref{alg:cap} is dominated by the $\mathcal{O}(N^3)$ nuclear norm evaluation. Across $|\mathcal{P}_u||\mathcal{P}_e|K$ constraints from C1, this yields a total complexity of $\mathcal{O}(I_{\max}|\mathcal{P}_u||\mathcal{P}_e|N^3)$.

Conversely, Algorithm~\ref{alg:cap 2} exhibits a drastically lower complexity by exploiting the \gls{LOS}-dominant \gls{mmWave} regime, which restricts the rank of $\mathbf{A}_{u,e}(\gamma)$ to at most two. Instead of full eigenvalue decomposition in \eqref{eq: calculation matrix Phi}, the problem reduces to finding the roots of the characteristic polynomial for a $2 \times 2$ matrix. Because the spatial inner products in this matrix are independent of the iterative variables ($\gamma$ and $\bs(T)$), they can be precomputed. This reduces the dominant per-iteration operation to a simple quadratic equation, scaling the total complexity of Algorithm~\ref{alg:cap 2} linearly with $N$ to $\mathcal{O}(J_{\max}I_{\max}|\Pset_u||\Pset_e|N)$.

\section{Extension to Temperature-Robust Phase-Shift Design}
\label{sec: Robust Design}
While the temperature-adaptive methods presented in Section~\ref{sec: Temperature-Aware LC-RIS Phase-shift Design} provide a highly secure rate for legitimate users, they rely on the premise that the instantaneous operating temperature $T$ of the \gls{LC}-\gls{RIS} is perfectly known at the controller. This can be achieved by equipping the \gls{RIS} with an array of thermal sensors and continuously feeding this data back to the \gls{BS}. However, this may introduce hardware complexity and signaling overhead. To circumvent this, we propose a temperature-robust phase-shift design. The objective is to compute a single, static nominal phase-shift configuration, denoted as $\bs_0$, that guarantees a high secure data rate across a predefined range of potential operating temperatures, thereby eliminating the need for real-time thermal tracking.

Let $\Tset = [T_{\min}, T_{\max}]$ denote the set of all possible operating temperatures\footnote{We can define such a range for each country separately.}. Based on the physical \gls{LC} response characteristics detailed in Section~\ref{sec: LC phase shifter}, applying a nominal phase-shift vector $\bs_0$ for the reference temperature ($T_r$) at the controller results in an actual, temperature-dependent reflection vector $\bs(T) = \mathcal{F}(\bs_0, T)$, where $\mathcal{F}(\cdot)$ captures the thermal phase drift based on \eqref{eq: omega in temperature final}. As we assumed only \gls{RIS} is temperature-dependent and not \gls{BS}, the assumptions for the \gls{BS} remain the same, and we exploit Lemma~\ref{lemma beamforming} to optimize the \gls{BS} without any change. In addition, as discussed in Section~\ref{sec: Algorithm and complexity analysis}, the complexity of the \gls{SDP}-based algorithm is much more than the low-complexity method so we only explain the robust version of the scalable method due to the space constraint\footnote{With the same line of thought, the \gls{SDP}-based algorithm can be also extended into a temperature-robust version.}. Therefore, we start by modifying a version of the problem P6 and extending it into a robust problem. To ensure robust, secure communication, we formulate a worst-case optimization problem that jointly considers the spatial uncertainty of the users and the thermal uncertainty of the hardware:
\begin{subequations}
\label{eq:optimization_robust}
\begin{align}
    \text{P9:}&~\underset{\gamma,\bs_0}{\max}~\gamma
    \\&~\text {s.t.} ~~\text{C1: }\frac{1+\SNR_u(\bs(T))}{1+\SNR_e(\bs(T))}\geq\gamma, \nonumber\\
    &\qquad\qquad\forall (\bp_u,\bp_e, T) \in \Pset_u \times \Pset_e \times \Tset,
    \\&\quad\hphantom {\text {s.t.} }\text{C2: } \bs(T) = \mathcal{F}(\bs_0, T), \quad \forall T \in \Tset,
    \\&\quad\hphantom {\text {s.t.} }\text{C3: } |[\bs_0]_n|=1,\,\forall n.
\end{align}
\end{subequations}
Problem~P9 is more challenging because the worst-case secrecy rate must be maximized over a strictly larger, three-dimensional uncertainty space (legitimate user location, eavesdropper location, and temperature). 

To solve this efficiently without increasing the computational complexity too much, we extend the \gls{LSE} surrogate framework introduced in Lemma~\ref{lem: LSE}. Rather than evaluating the surrogate bounds solely over the spatial domain, we discretize the temperature uncertainty set $\Tset$ into $M$ representative sample temperatures, $\Tset_d = \{T_1, T_2, \dots, T_M\}$. We then evaluate the secrecy rate $\widetilde{\RS}(\bp_u,\bp_e, T_m)$ for all possible combinations of locations and temperatures simultaneously. 

By applying the \gls{LSE} approximation over this joint discrete space, the spatial-thermal weights are defined as:
\begin{subequations}
    \begin{align}
    \label{eq: robust weights}
    &w(\bp_u,\bp_e, T_m)\defeq\exp\left(-\frac{\widetilde{\RS}(\bp_u,\bp_e, T_m)}{\mu}\right),\\
    &w_\text{total}\defeq\sum_{T_m\in\Tset_d}\sum_{\bp_u\in\Pset_U}\sum_{\bp_e\in\Pset_e}w(\bp_u,\bp_e, T_m).
    \end{align}
\end{subequations}
This joint weighting mechanism, including the temperature range, is the key of the robust design. If a specific spatial configuration becomes highly vulnerable to eavesdropping only when the \gls{RIS} drops to a specific temperature (e.g., $T_m = 0^\circ\text{C}$), the \gls{LSE} formulation automatically assigns an exponentially larger weight to that specific $( \bp_u, \bp_e, T_m )$ item. 

Consequently, the aggregated unconstrained phase gradient incorporates the vulnerabilities from all thermal states. The robust nominal phase angles can be iteratively computed as:
\begin{equation}
\label{eq:LSE_robust}
\bomega_0 = \arg\!\left(\sum_{T_m\in\Tset_d}\sum_{\bp_u\in\Pset_U}\sum_{\bp_e\in\Pset_e}\!\!\!\frac{w(\bp_u,\bp_e, T_m)}{w_\text{total}}\bbeta_{u,e}(\gamma, T_m)\!\right)\!\!,
\end{equation}
where $\bbeta_{u,e}(\gamma, T_m)$ is the temperature-specific surrogate reference vector evaluated using $\tilde{\bs}(T_m)$. Finally, to ensure feasibility at the controller, $\bomega_0$ is projected into the valid hardware tuning range automatically by \eqref{eq: omega in temperature final} at each specific temperature. The resulting configuration $\bs_0^\star$ provides an intrinsically robust baseline that safeguards the secrecy rate against unpredictable environmental fluctuations.

\section{Performance Evaluation}
\label{sec: Performance Comparison}
\subsection{Simulation Setup}
\label{sec: Simulation Setup}
We employ the simulation configuration for coverage extension presented in Fig. \ref{fig:system model}, where the \gls{RIS} center is the origin of the Cartesian coordinate system, i.e., $[0,0,0]~\text{m}$. We assume there is a legitimate user in a fixed area $\Pset_u\in\{(\x,\y,\z):4~\text{m}\leq\x\leq 6~\text{m}, -3~\text{m}\leq\y\leq -1~\text{m}, \z=-5\}$. The \gls{BS} comprises a $16\times16=256$ \gls{UPA} positioned along the 
$\x-\z$ plane, and located at $[30,0,0]~\text{m}$. The \gls{RIS} is a \gls{UPA} consisting of $N_\y\times N_\z=20\times20$ elements aligned to the $\y$ and $\z$ axes, respectively, unless explicitly mentioned otherwise. The element space for both the \gls{BS} and \gls{RIS} is half of the wavelength. The noise variance is computed as $\sigma_n^2=WN_0N_{\rm f}$ with $N_0=-174$~dBm/Hz, $W=100$~MHz, and $N_{\rm f}=6$~dB. We assume $60$~GHz carrier frequency, and $\rho(d_0/d)^\sigma$ pathloss model where $\rho=-68$~dB at $d_0=1$~m. Moreover, we adopt the pathloss exponent $\sigma = (2,2,2)$ and Ricean $K$-factors in \eqref{eq: channel model}, $\bar{k}_r=\tilde{k}_r=(0,0.1,0.1),\,\forall r$, for the \gls{BS}-\gls{MU}, \gls{BS}-\gls{RIS}, and \gls{RIS}-\gls{MU} channels, respectively, and $R=10$. 

The analysis considers two scenarios based on the location of the eavesdropper \gls{w.r.t.} the legitimate user area. These scenarios are critical because the eavesdropper's location affects how the \gls{RIS} can optimize secure rate communication.
\begin{itemize}
    \item \textbf{Scenario 1: Same Distance, Different Angles:} Here, the eavesdropper is positioned away from the legitimate user at different angles seen from the \gls{RIS}. The eavesdropper area is at a different y-coordinate but the same x-coordinate as the legitimate user.
    \item \textbf{Scenario 2: Different Distance, Same Angle:} In this case, the eavesdropper is positioned at the same angle as the legitimate user but with a  different distance from the \gls{RIS}. The eavesdropper area is positioned below the legitimate user area.
\end{itemize}
Furthermore, two approaches are considered in managing the \gls{RIS} phase shifts: \textit{Neglecting} temperature changes, where no adjustment is made, and \textit{Optimizing} the phase shifts to account for temperature variations to ensure secure communication. The other parameters used in the simulations are as follows: $\beta=0.25$, $T_c=95~^{\circ}$C, $T_r=10~^{\circ}$C, $P_t=37~$dBm, $\eta^{(0)}=0.2$, $I_{\max}=15$, $J_{\max}=6$, $\epsilon_1=0.01$, $\epsilon_2=0.1$, and $\gamma^{(0)}=1$.
 \begin{figure}
    \centering
    \includegraphics[width=0.5\textwidth]{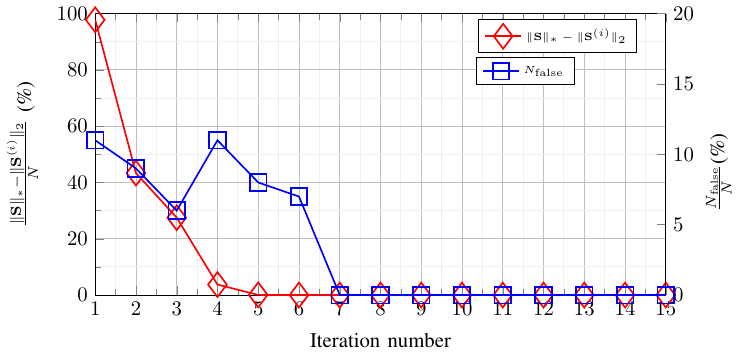}
    \caption{The percentage of $\|\bS\|_*-\|\bS^{(i)}\|_2$ and $N_\mathrm{false}$ out of total number of $N$ at each iteration, when $T=40^\circ$C.}
    \label{fig: rank}
    \vspace{-5mm}
\end{figure}

\begin{figure}
    \centering
    \includegraphics[width=0.5\textwidth]{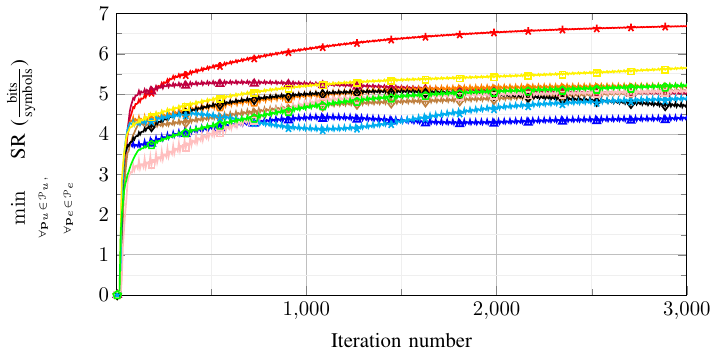}
    \caption{Convergence behavior of Algorithm~\ref{alg:cap 2} with different initializations versus the iteration number, when $T=40^\circ$C.}
    \label{fig: convergence}
    \vspace{-5mm}
\end{figure}

 \begin{figure}
    \centering
    \includegraphics[width=0.5\textwidth]{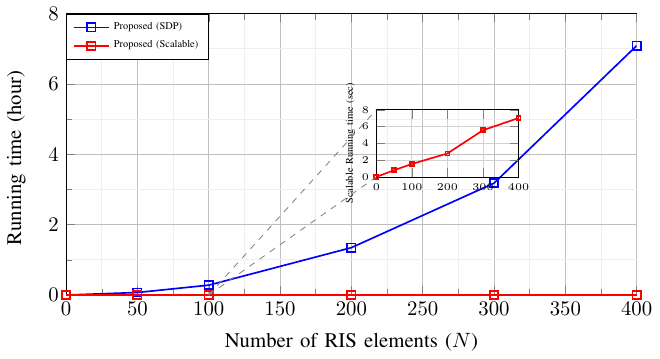}
    \caption{Running time of both proposed algorithms versus the number of the
\gls{RIS} elements ($N$).}
    \label{fig: running time}
    \vspace{-5mm}
\end{figure}

\begin{figure}
    \centering
    \includegraphics[width=0.5\textwidth]{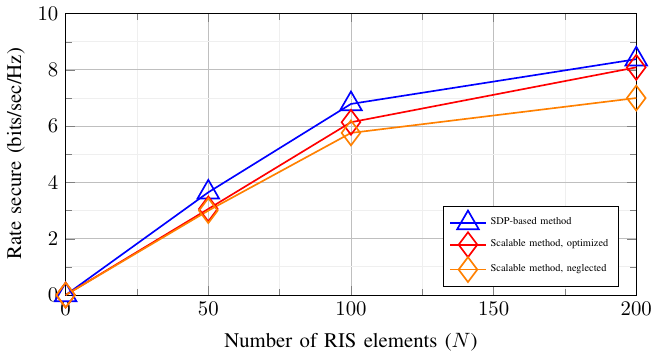}
    \caption{The minimum secure rate over all possible locations $\bp_u\in\Pset_u$, $\bp_e\in\Pset_e$ versus different number of \gls{RIS} elements when $T=40^\circ$C in scenario 1.}
    \label{fig:SR_N}
    \vspace{-5mm}
\end{figure}

\begin{remk}
The MATLAB codes used to generate the simulation results in this section are publicly available online at\\ \href{https://github.com/MohamadrezaDelbari/LC-RIS-temperature}{\textcolor{blue}{https://github.com/MohamadrezaDelbari/LC-RIS-temperature}}.
\end{remk}

\subsection{Simulation Result}
\subsubsection{Convergence and complexity comparison}
The convergence characteristics of Algorithm \ref{alg:cap} are illustrated in Fig.~\ref{fig: rank}. In this context, $N_\mathrm{false}$ represents the count of elements that exceed the phase shift threshold $\omega_{\max}(T)$ (in this figure $T=40^\circ$C), thereby violating constraint C2. As the iterations progress, both the normalized nuclear-to-spectral norm difference $\|\bS\|_* - \|\bS^{(i)}\|_2$ and $N_\mathrm{false}$ diminish toward zero percent. This trend confirms that the rank-one requirement for $\bS$ and the conditions of C2 are successfully satisfied.

Fig. \ref{fig: convergence} displays how Algorithm~\ref{alg:cap 2} converges across 10 different random starting points, tracking the minimum \gls{RS} against the iteration count. Most initializations successfully increase the secrecy rate. Because \eqref{eq: phase thresholding} yields a sub-optimal instead of a global solution to Problem P6, a strictly monotonic increase in \gls{RS} is not guaranteed, but the algorithm reliably converges to a high-quality local optimum (indicated by the red curve) that substantially boosts the overall secrecy rate.

Fig.~\ref{fig: running time} compares the execution times of the two proposed algorithms: the \gls{SDP}-based approach and the scalable method\footnote{The algorithms were implemented in MATLAB R2024a and executed on an Arch Linux system equipped with an AMD Ryzen 9 7950X (16-core) CPU and 64 GB of RAM.}. As illustrated, the scalable method (seconds) requires significantly less computational time than the \gls{SDP} method (hours). As anticipated, the computational complexity of the \gls{SDP} approach scales cubically, approximately $\mathcal{O}(N^3)$, whereas the scalable solution exhibits a linear growth rate of $\mathcal{O}(N)$ approximately. For instance, at $N=400$, the \gls{SDP} approach requires approximately 7 hours to complete, while the scalable approach takes only 7 seconds. Note that these figures represent a single initialization for both algorithms. Our observations indicate that the scalable approach may require dozens of initializations to achieve optimal results; however, this only increases the total runtime to approximately one minute, which is highly practical for real-time applications when $N=400$.

\subsubsection{Performance comparison}

Figure~\ref{fig:SR_N} illustrates the minimum secrecy rate across all potential user and eavesdropper locations for \(T=40\), comparing Algorithm~\ref{alg:cap} and Algorithm~\ref{alg:cap 2} (optimized and neglected benchmark). While the \gls{SDP} method consistently outperforms the scalable approach, this performance gain comes at the expense of higher computational complexity, as demonstrated in Fig.~\ref{fig: running time}. Consequently, the remaining evaluations focus exclusively on Algorithm~\ref{alg:cap 2} due to its superior execution speed when \(N=400\).


\begin{figure*}[ht]
\centering
\begin{subfigure}{0.19\textwidth}
    \caption{Neglected, $T=-20^\circ$C}
    \includegraphics[width=\textwidth,height=0.7\textwidth]{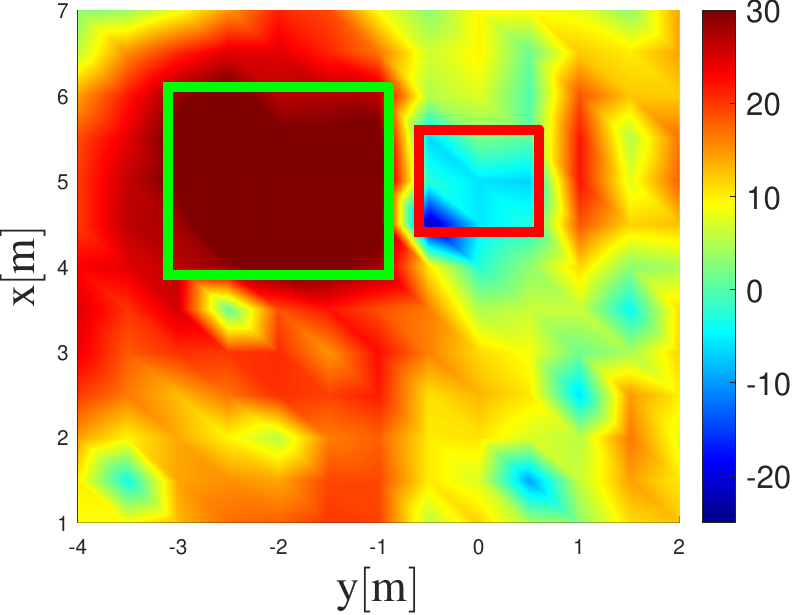}
    \label{fig: T-20_Scalable_changed_U5-2_E50}
\end{subfigure}\hfill
\begin{subfigure}{0.19\textwidth}
    \caption{Neglected, $T=-10^\circ$C}
    \includegraphics[width=\textwidth,height=0.7\textwidth]{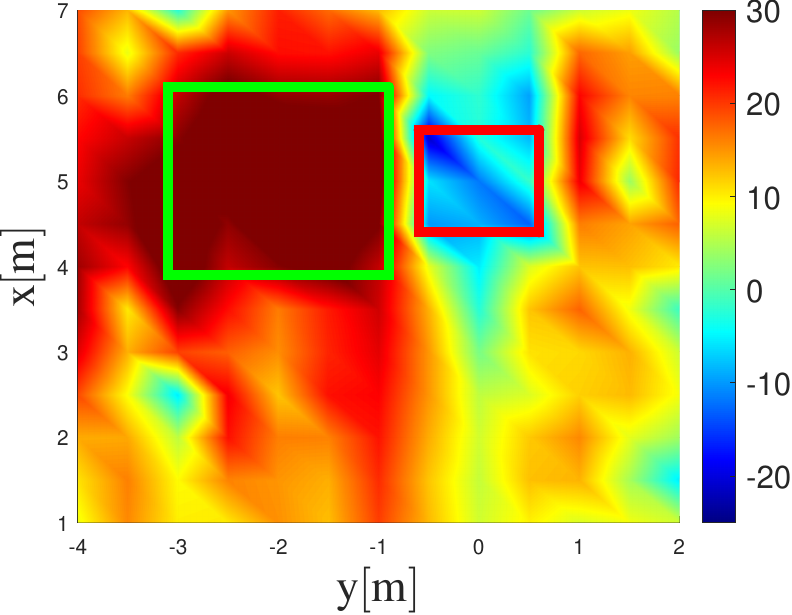}
    \label{fig: T-10_Scalable_changed_U5-2_E50}
\end{subfigure}\hfill
\begin{subfigure}{0.19\textwidth}
    \caption{Neglected, $T=10^\circ$C}
    \includegraphics[width=\textwidth,height=0.7\textwidth]{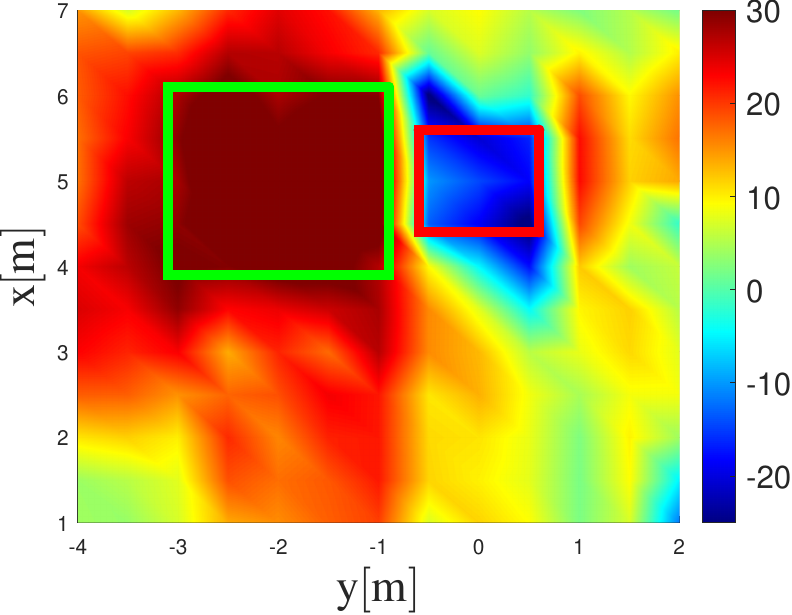}
    \label{fig: T10_Scalable_changed_U5-2_E50}
\end{subfigure}\hfill
\begin{subfigure}{0.19\textwidth}
    \caption{Neglected, $T=30^\circ$C}
    \includegraphics[width=\textwidth,height=0.7\textwidth]{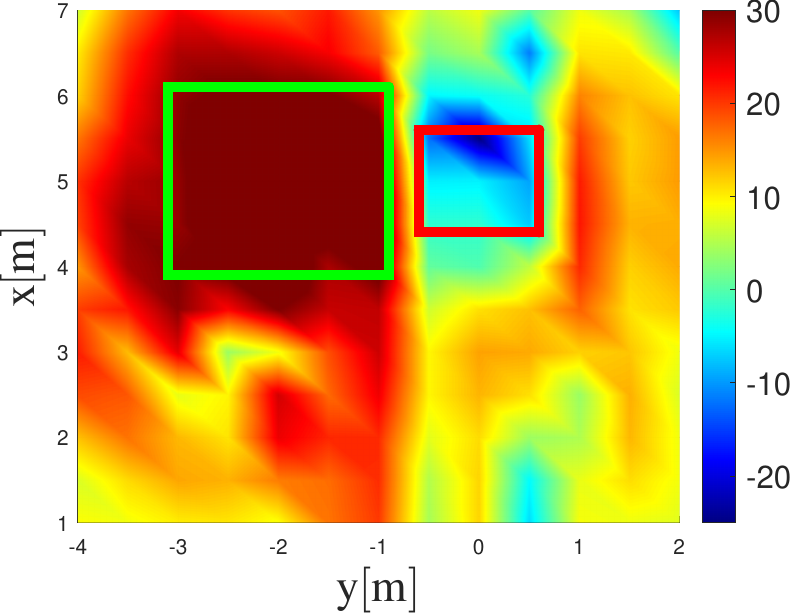}
    \label{fig: T30_Scalable_changed_U5-2_E50}
\end{subfigure}\hfill
\begin{subfigure}{0.19\textwidth}
    \caption{Neglected, $T=40^\circ$C}
    \includegraphics[width=\textwidth,height=0.7\textwidth]{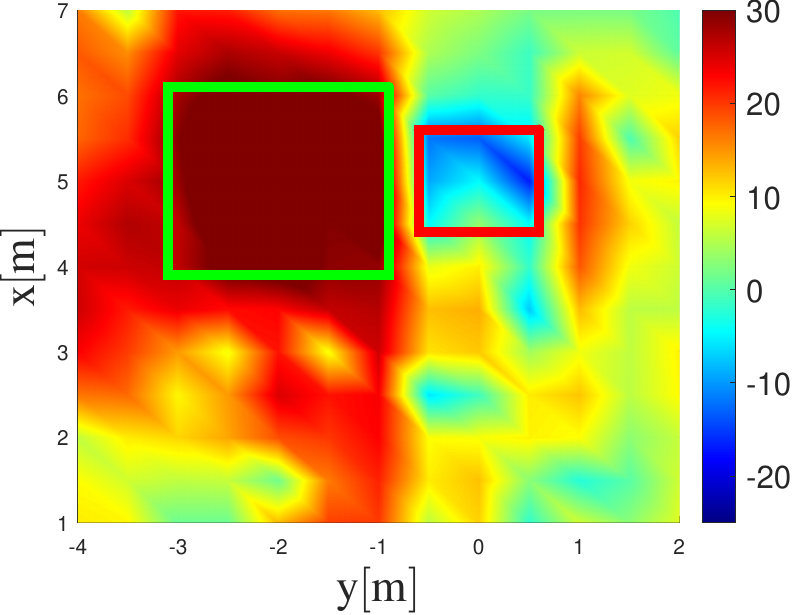}
    \label{fig: T40_Scalable_changed_U5-2_E50}
    \end{subfigure}\hfill
    \vspace{-5mm}
    \newline
    \begin{subfigure}{0.19\textwidth}
    \caption{Optimized, $T=-20^\circ$C}
    \includegraphics[width=\textwidth,height=0.7\textwidth]{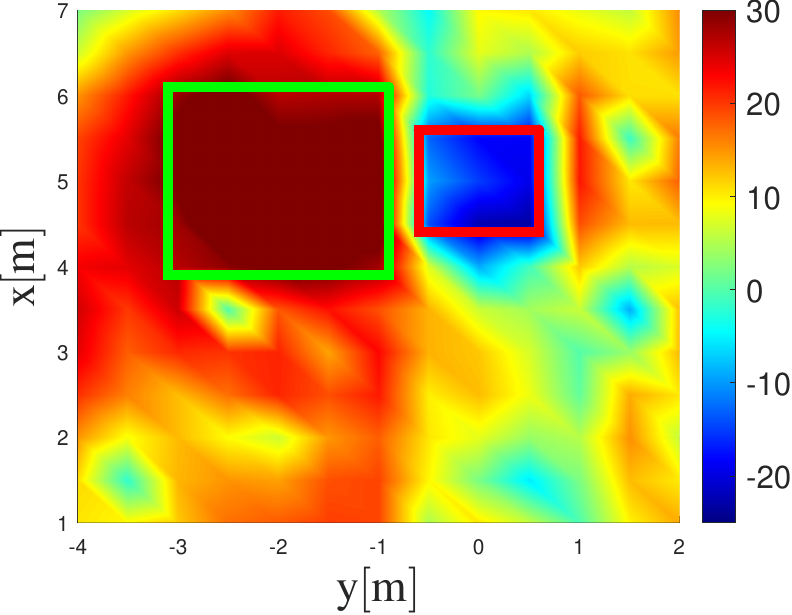}
    \label{fig: T-20_Scalable_Opt_U5-2_E50}
\end{subfigure}\hfill
\begin{subfigure}{0.19\textwidth}
    \caption{Optimized, $T=-10^\circ$C}
    \includegraphics[width=\textwidth,height=0.7\textwidth]{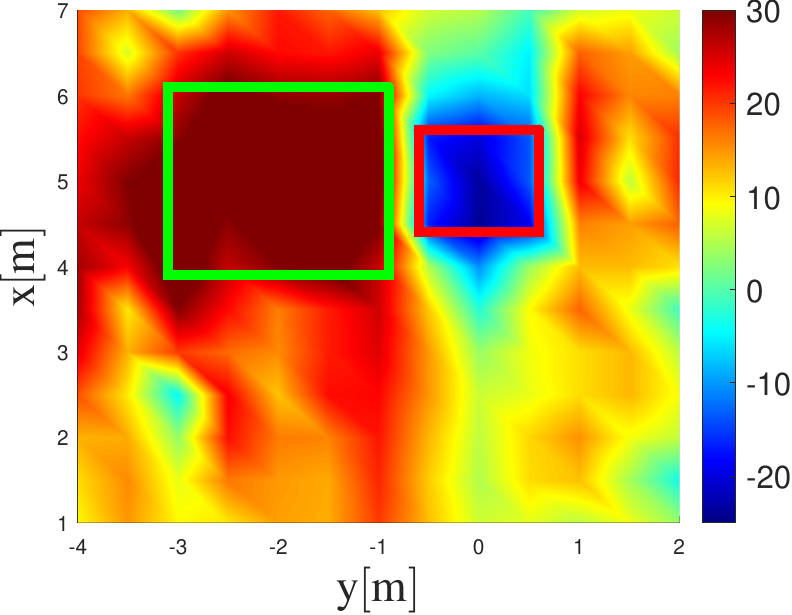}
    \label{fig: T-10_Scalable_Opt_U5-2_E50}
\end{subfigure}\hfill
\begin{subfigure}{0.19\textwidth}
    \caption{Optimized, $T=10^\circ$C}
    \includegraphics[width=\textwidth,height=0.7\textwidth]{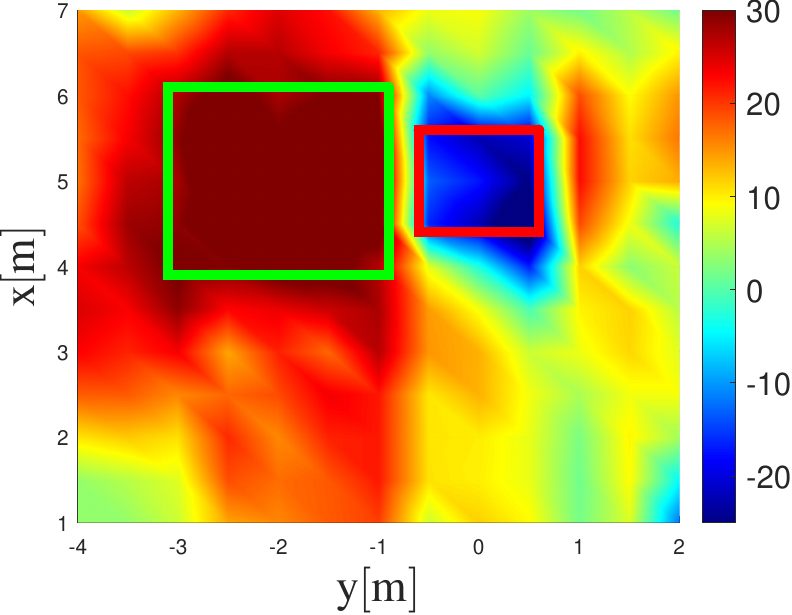}
    \label{fig: T10_Scalable_Opt_U5-2_E50}
\end{subfigure}\hfill
\begin{subfigure}{0.19\textwidth}
    \caption{Optimized, $T=30^\circ$C}
    \includegraphics[width=\textwidth,height=0.7\textwidth]{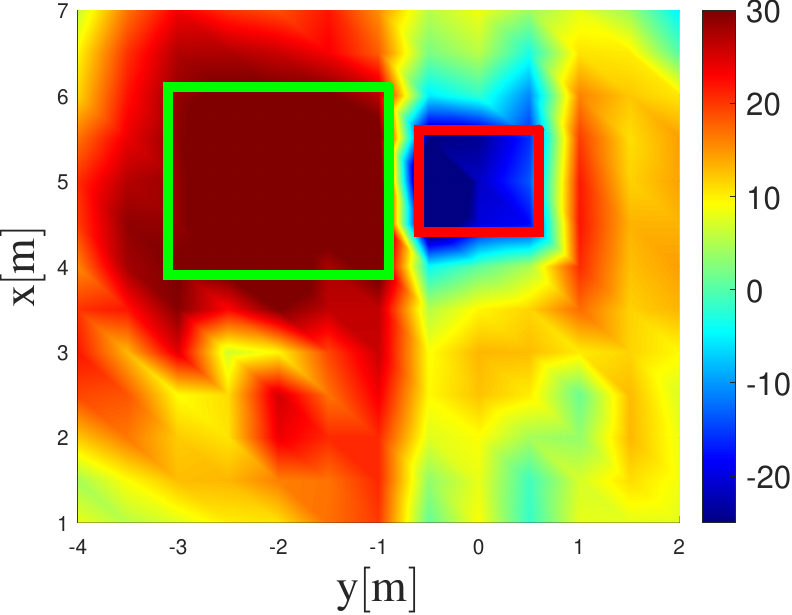}
    \label{fig: T30_Scalable_Opt_U5-2_E50}
\end{subfigure}\hfill
\begin{subfigure}{0.19\textwidth}
    \caption{Optimized, $T=40^\circ$C}
    \includegraphics[width=\textwidth,height=0.7\textwidth]{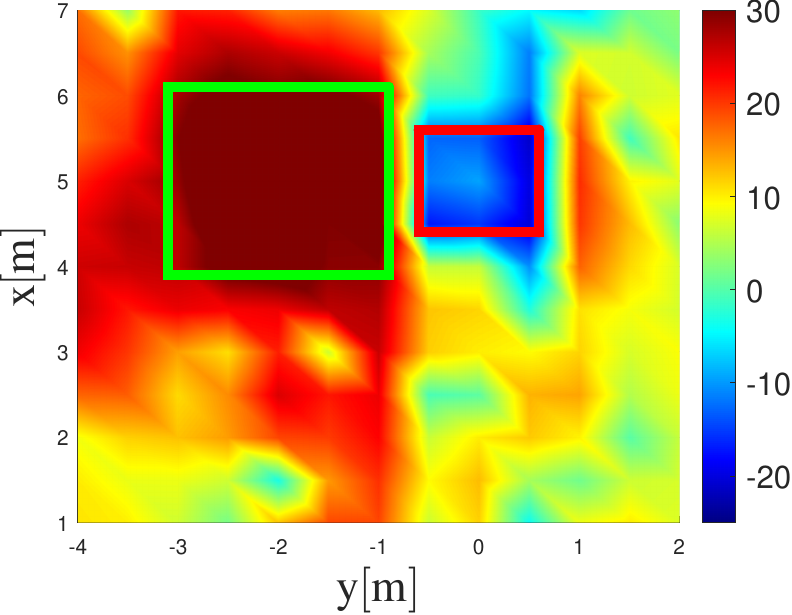}
    \label{fig: T40_Scalable_Opt_U5-2_E50}
\end{subfigure}\hfill
\vspace{-5mm}
\caption{SNR (dB) for two cases in scenario 1. The first row includes results when the \gls{RIS} phase shifts are designed without caring about the temperature, while the bottom row demonstrates results for the case where the \gls{RIS} phase shifts are optimized based on the temperature.}
\label{fig: heat map H}
\vspace{-5mm}
\end{figure*}


\begin{figure*}[ht]
\centering
\begin{subfigure}{0.19\textwidth}
    \caption{Neglected, $T=-20^\circ$C}
    \includegraphics[width=\textwidth,height=0.7\textwidth]{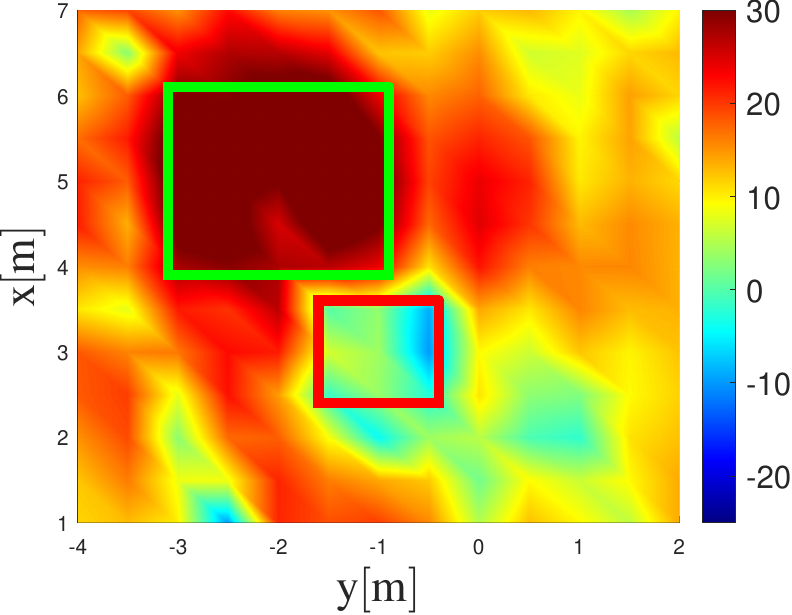}
    \label{fig: T-20_Scalable_changed_U5-2_E3-1}
\end{subfigure}\hfill
\begin{subfigure}{0.19\textwidth}
    \caption{Neglected, $T=-10^\circ$C}
    \includegraphics[width=\textwidth,height=0.7\textwidth]{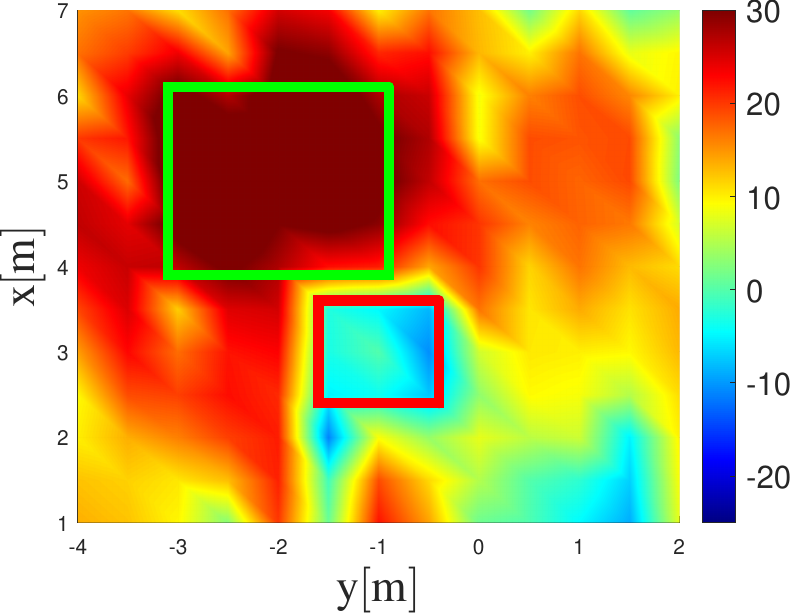}
    \label{fig: T-10_Scalable_changed_U5-2_E3-1}
\end{subfigure}\hfill
\begin{subfigure}{0.19\textwidth}
    \caption{Neglected, $T=10^\circ$C}
    \includegraphics[width=\textwidth,height=0.7\textwidth]{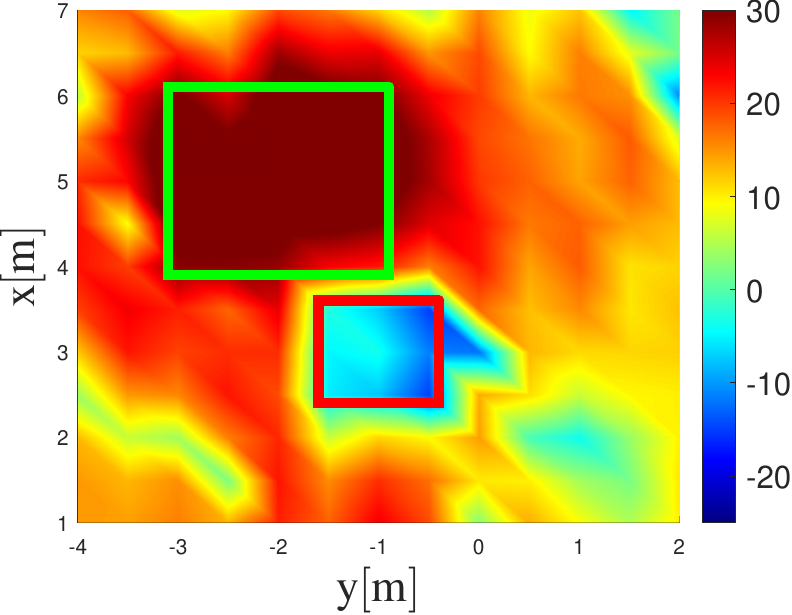}
    \label{fig: T10_Scalable_changed_U5-2_E3-1}
\end{subfigure}\hfill
\begin{subfigure}{0.19\textwidth}
    \caption{Neglected, $T=30^\circ$C}
    \includegraphics[width=\textwidth,height=0.7\textwidth]{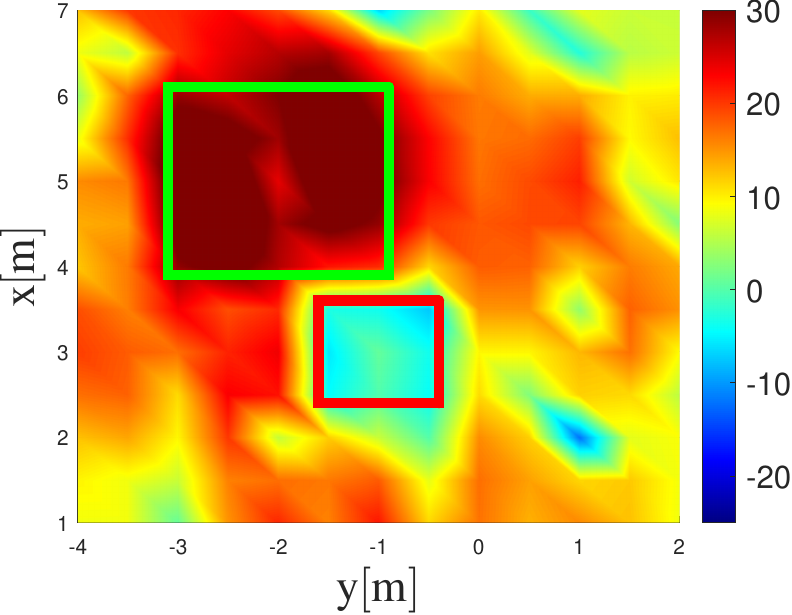}
    \label{fig: T30_Scalable_changed_U5-2_E3-1}
\end{subfigure}\hfill
\begin{subfigure}{0.19\textwidth}
    \caption{Neglected, $T=40^\circ$C}
    \includegraphics[width=\textwidth,height=0.7\textwidth]{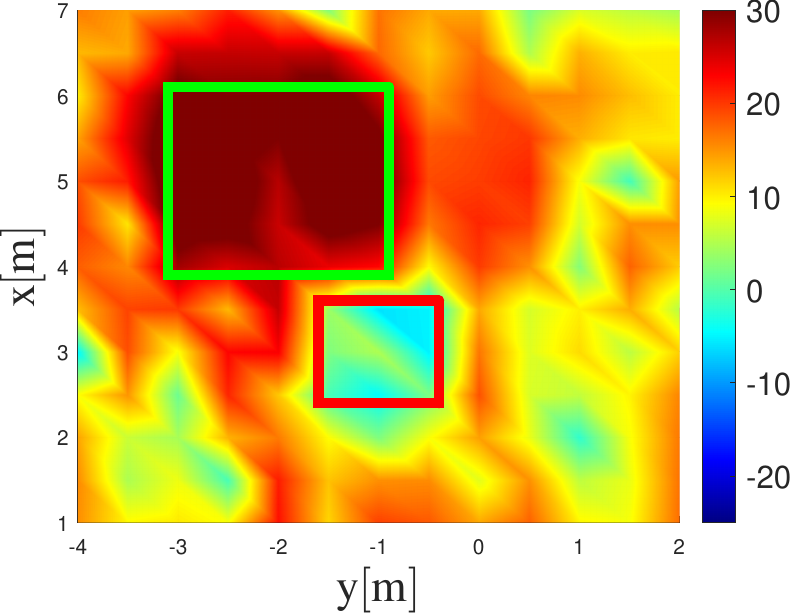}
    \label{fig: T40_Scalable_changed_U5-2_E3-1}
    \end{subfigure}\hfill
    \vspace{-5mm}
    \newline
    \begin{subfigure}{0.19\textwidth}
    \caption{Optimized, $T=-20^\circ$C}
    \includegraphics[width=\textwidth,height=0.7\textwidth]{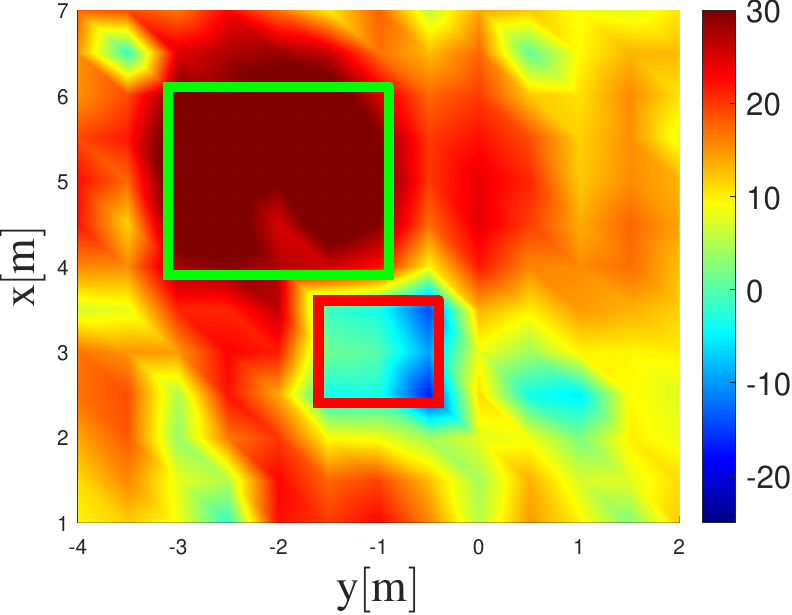}
    \label{fig: T-20_Scalable_Opt_U5-2_E3-1}
\end{subfigure}\hfill
\begin{subfigure}{0.19\textwidth}
    \caption{Optimized, $T=-10^\circ$C}
    \includegraphics[width=\textwidth,height=0.7\textwidth]{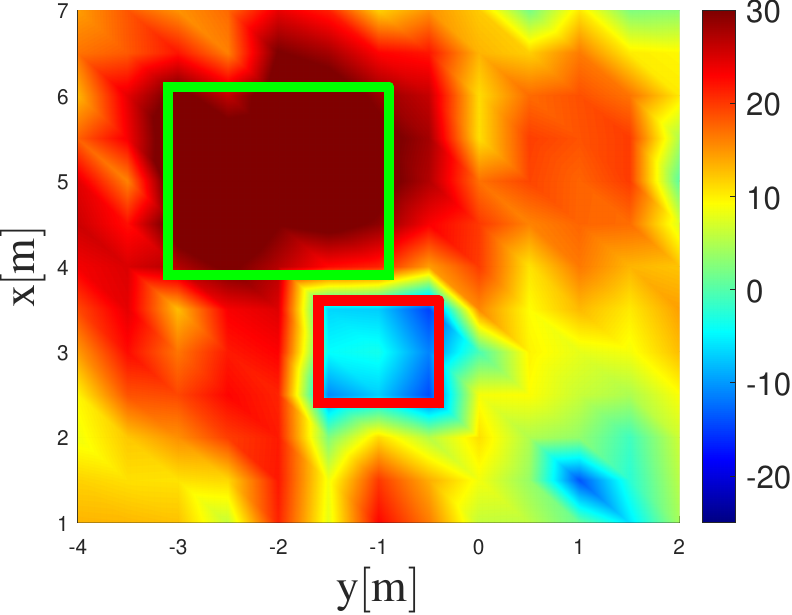}
    \label{fig: T-10_Scalable_Opt_U5-2_E3-1}
\end{subfigure}\hfill
\begin{subfigure}{0.19\textwidth}
    \caption{Optimized, $T=10^\circ$C}
    \includegraphics[width=\textwidth,height=0.7\textwidth]{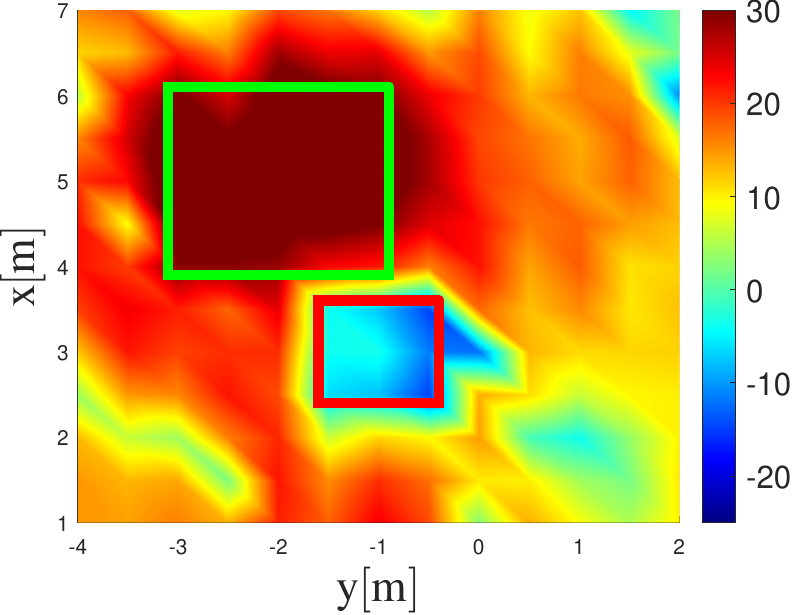}
    \label{fig: T10_Scalable_Opt_U5-2_E3-1}
\end{subfigure}\hfill
\begin{subfigure}{0.19\textwidth}
    \caption{Optimized, $T=30^\circ$C}
    \includegraphics[width=\textwidth,height=0.7\textwidth]{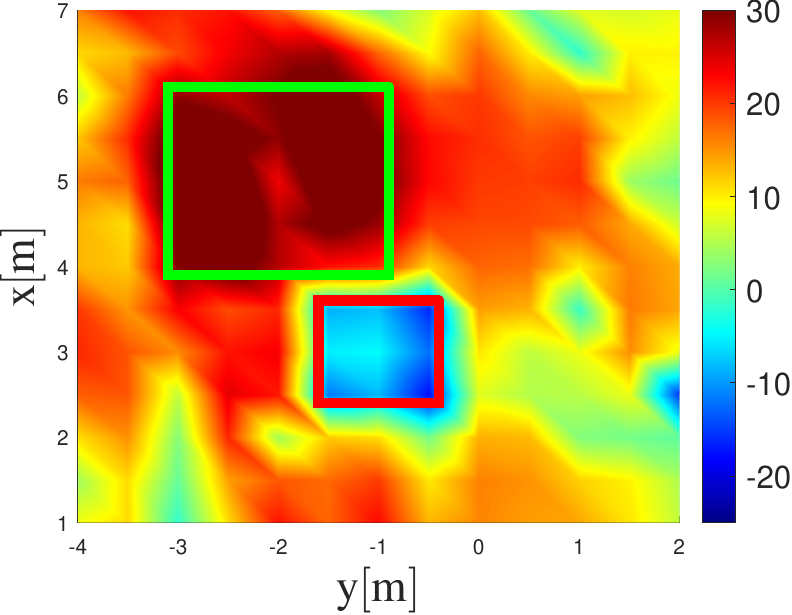}
    \label{fig: T30_Scalable_Opt_U5-2_E3-1}
\end{subfigure}\hfill
\begin{subfigure}{0.19\textwidth}
    \caption{Optimized, $T=40^\circ$C}
    \includegraphics[width=\textwidth,height=0.7\textwidth]{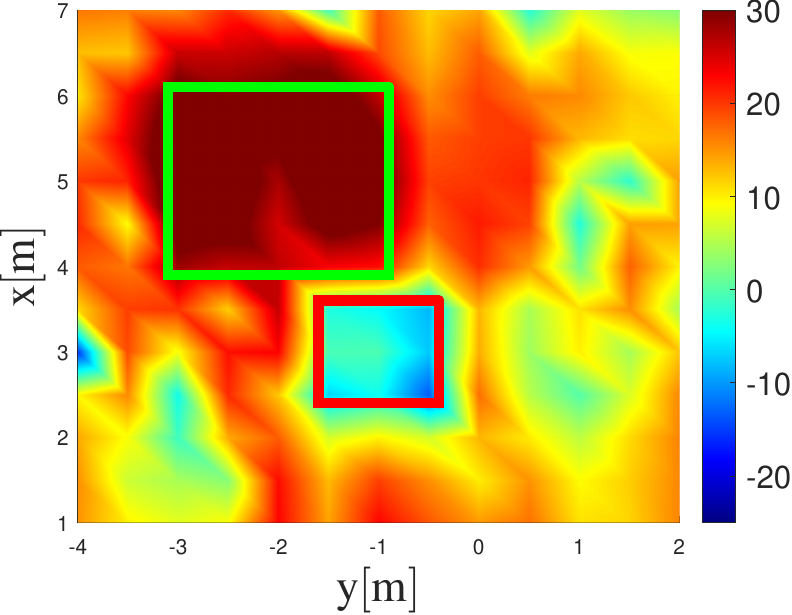}
    \label{fig: T40_Scalable_Opt_U5-2_E3-1}
\end{subfigure}\hfill
\vspace{-5mm}
\caption{SNR (dB) for two cases in scenario 2. The first row includes results when the \gls{RIS} phase shifts are designed without caring about the temperature, while the bottom row demonstrates results for the case where the \gls{RIS} phase shifts are optimized based on the temperature.}
\label{fig: heat map V}
\vspace{-5mm}
\end{figure*}

Figs. \ref{fig: heat map H} and \ref{fig: heat map V} illustrate the averaged received \gls{SNR} (dB) at various locations for Scenarios 1 and 2, respectively, via the proposed scalable algorithm. In both figures, the columns represent different temperatures ranging from $-20^\circ$C to $40^\circ$C, with $10^\circ$C serving as the reference temperature. The first row in each figure depicts the performance when temperature effects are neglected, while the second row shows the performance when \gls{RIS} phase shifts are optimized to account for these thermal impacts. While the received \gls{SNR} for \gls{MU} appears stable if temperature effects are ignored, the signal strength within the \gls{ME}’s vicinity increases considerably under those conditions, compromising security. Furthermore, as the operating temperature deviates from the reference temperature, the negative impact of neglecting thermal variations becomes more pronounced. In general, maximizing the secure rate is more straightforward in Scenario 1 (Fig. \ref{fig: heat map H}) because the \gls{MU} area and the \gls{ME} are located at different angles. In contrast, Scenario 2 (Fig. \ref{fig: heat map V}) presents a greater challenge; however, a high secure rate is still achievable due to the additional degrees of freedom provided by operating in the \gls{NF} regime.

\begin{figure}
    \centering
    \includegraphics[width=0.5\textwidth]{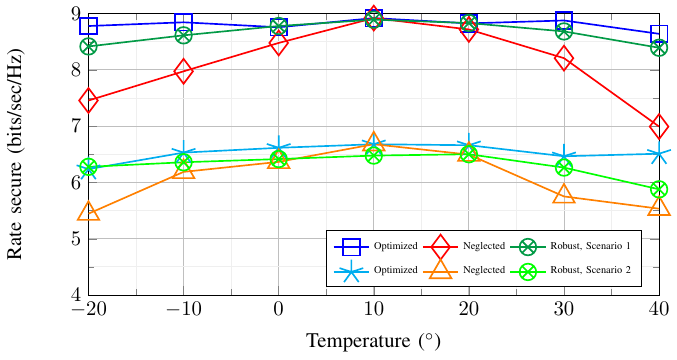}
    \caption{The minimum secrecy rate (bits/sec/Hz) over all possible locations $\bp_u\in\Pset_u$ and $\bp_e\in\Pset_e$ versus temperature.}
    \label{fig: SR_T}
    \vspace{-5mm}
\end{figure}

In Fig.~\ref{fig: SR_T}, we plot the secure rate versus temperature for both scenarios, comparing cases where the temperature impact is either accounted for or neglected. As illustrated in both scenarios, neglecting temperature effects causes the secure rate to drop as the operating temperature deviates from the reference ($T_r=10^\circ$C). Conversely, the secure rate remains relatively constant across the entire range when the \gls{LC}-\gls{RIS} phase shifts are optimized according to the actual temperature. Overall, scenario 1 achieves a higher secure rate than scenario 2 because the \gls{ME} and the \gls{MU} regions are not aligned in the same direction relative to the \gls{RIS}. In addition, we plotted the robust algorithm result derived in \eqref{eq:LSE_robust}. It has approximately a constant secure rate across all temperatures, although the exact temperature data is not available.


\section{Conclusion}
\label{sec: Conclusion}
In this paper, we first analyzed the impact of temperature variations on the phase shifts of \gls{LC}-\glspl{RIS}. To mitigate the phase shift range limitations caused by rising temperatures, we developed two algorithms aimed at maximizing the secrecy rate. Furthermore, we proposed a robust design that optimizes the \gls{LC}-\gls{RIS} phase shifts without requiring prior knowledge of the temperature. Simulation results demonstrate the critical importance of accounting for temperature variations in \gls{LC}-\gls{RIS} phase shift design. As future work, one may explore signal processing and machine learning approaches to estimate the temperature from the feedback signal by the user, thereby optimizing phase shifts for enhanced security.

\bibliographystyle{IEEEtran}
\bibliography{References}
\end{document}

%% file: preamble.tex
\usepackage{amsfonts,amsmath,amsthm,bm}
\usepackage{cite}
\usepackage{amsmath,amssymb,amsfonts}
\usepackage{algorithmic}
\usepackage{graphicx}
\usepackage{textcomp}
\usepackage{xcolor}
\usepackage[T1]{fontenc}
\usepackage{comment}
\usepackage{eucal}
\usepackage{algorithm}
\usepackage{algorithmic}
\usepackage[acronyms,nonumberlist,nopostdot,nomain,nogroupskip]{glossaries}
\usepackage[font=footnotesize]{caption}
\usepackage{subcaption}

\usepackage{booktabs,flushend}
\usepackage{stfloats}
\usepackage{microtype}

\newtheorem{lem}{Lemma}

\newtheorem{remk}{Remark}

\newcommand{\defeq}{\triangleq}
\def\bOmega{\boldsymbol{\Omega}}

\newcommand{\e}{\mathsf{e}}
\newcommand{\jj}{\mathsf{j}}
\newcommand{\Herm}{\mathsf{H}}
\newcommand{\Trans}{\mathsf{T}}
\newcommand{\x}{\mathsf{x}}
\newcommand{\y}{\mathsf{y}}
\newcommand{\z}{\mathsf{z}}
\newcommand{\bA}{\mathbf{A}}
\newcommand{\bx}{\mathbf{x}}

\newcommand{\bs}{\mathbf{s}}
\newcommand{\bS}{\mathbf{S}}
\newcommand{\bn}{\mathbf{n}}
\newcommand{\bH}{\mathbf{H}}
\newcommand{\ba}{\mathbf{a}}

\newcommand{\bu}{\mathbf{u}}

\newcommand{\bI}{\mathbf{I}}

\newcommand{\bh}{\mathbf{h}}

\newcommand{\bq}{\mathbf{q}}
\newcommand{\bp}{\mathbf{p}}

\newcommand{\kk}{\kappa}

\newcommand{\bSigma}{\boldsymbol{\Sigma}}
\newcommand{\bGamma}{\boldsymbol{\Gamma}}
\newcommand{\bPhi}{\boldsymbol{\Phi}}
\newcommand{\bbeta}{\boldsymbol{\beta}}
\newcommand{\bmu}{\boldsymbol{\mu}}

\newcommand{\blambda}{\boldsymbol{\lambda}}

\newcommand{\Ex}{\mathbb{E}}

\newcommand{\diag}{\mathrm{diag}}
\newcommand{\real}{\mathrm{Re}}
\newcommand{\imag}{\mathrm{Im}}
\newcommand{\tr}{\mathrm{tr}}
\newcommand{\rank}{\mathrm{rank}}
\newcommand{\dd}{\mathrm{d}}

\newcommand{\tx}{\mathrm{tx}}

\newcommand{\SNR}{\mathrm{SNR}}
\newcommand{\RS}{\mathrm{SR}}
\newcommand{\RIS}{\mathrm{RIS}}

\def\bomega{\boldsymbol{\omega}}

\def\bzero{\boldsymbol{0}}
\def\bone{\boldsymbol{1}}
\def\Cset{\mathbb{C}}

\def\Rset{\mathbb{R}}

\def\LOS{\mathrm{LOS}}

\def\tmax{\mathrm{max}}
\def\tx{\mathrm{tx}}
\def\rx{\mathrm{rx}}
\def\BS{\mathrm{BS}}
\def\RIS{\mathrm{RIS}}

\def\SNR{\mathrm{SNR}}
\def\eff{\mathrm{eff}}

\def\sCN{\mathcal{CN}}

\def\Pset{\mathcal{P}}
\def\Tset{\mathcal{T}}
\def\bigO{\mathcal{O}}
\usepackage{xcolor}

\input{acronyms}

%% file: acronyms.tex
\newacronym{RIS}{RIS}{reconfigurable intelligent surface}
\newacronym{QoS}{QoS}{quality of service}
\newacronym{LC}{LC}{liquid crystal}
\newacronym{SNR}{SNR}{signal to noise ratio}
\newacronym{TDMA}{TDMA}{time-division multiple-access}
\newacronym{BS}{BS}{base station}
\newacronym{MU}{MU}{mobile user}
\newacronym{ME}{ME}{mobile eavesdropper}
\newacronym{NF}{NF}{near-field}
\newacronym{Tx}{Tx}{transmitter}
\newacronym{Rx}{Rx}{receiver}
\newacronym{AWGN}{AWGN}{additive white Gaussian noise}
\newacronym{w.r.t.}{w.r.t.}{with respect to}
\newacronym{RDE}{RDE}{Reaction-Diffusion Equation}
\newacronym{PDE}{PDE}{partial differential equation}
\newacronym{UPA}{UPA}{uniform planar array}
\newacronym{AO}{AO}{alternative optimization}
\newacronym{SOCP}{SOCP}{second-order cone programming}
\newacronym{AoD}{AoD}{angle of departure}
\newacronym{LOS}{LOS}{line of sight}
\newacronym{nLOS}{NLOS}{non-LOS}
\newacronym{MIMO}{MIMO}{multiple-input multiple-output}
\newacronym{RS}{SR}{secure rate}
\newacronym{SDP}{SDP}{semi-definite programming}
\newacronym{CSI}{CSI}{channel state information}

\newacronym{PIN}{PIN}{positive-intrinsic-negative}
\newacronym{RF}{RF}{radio frequency}
\newacronym{MEMS}{MEMS}{micro-electro-mechanical system}
\newacronym{mmWave}{mmWave}{millimeter-wave}
\newacronym{OFDM}{OFDM}{orthogonal frequency division multiplexing}
\newacronym{LSE}{LSE}{log-sum-exp}
\newacronym{PDF}{PDF}{probability distribution function}